\documentclass{llncs}
\pdfoutput=1

\usepackage{algorithm, algpseudocode}

\usepackage{epsfig, endnotes, url, color, graphicx}
\usepackage{mathtools, amsmath, amsfonts}
\usepackage{diagbox, booktabs, colortbl}
\usepackage{breakurl, listings, framed}
\usepackage{multirow, caption, subcaption}
\usepackage[hidelinks]{hyperref}
\usepackage{xspace, slashbox, enumitem, bm}
\DeclareMathAlphabet{\mathcal}{OMS}{cmsy}{m}{n}

{
\makeatletter
\newcommand{\@chapapp}{\relax}%
\makeatother

\usepackage[title]{appendix}

\newcommand\sysname{\textsf{MicroCash}\xspace}
\newcommand\besc{$B_{\emph{escrow}}$\xspace}
\newcommand\bpen{$B_{\emph{penalty}}$\xspace}
\newcommand\bescind{$B_{\emph{escrow}}^{\emph{ind}}$\xspace}
\newcommand\bpenind{$B_{\emph{penalty}}^{\emph{ind}}$\xspace}
\newcommand\trefund{$t_{\emph{refund}}$\xspace}
\newcommand\tissue{$t_{\emph{issue}}$\xspace}
\newcommand\tdraw{$t_{\emph{draw}}$\xspace}

\newcommand\texpire{$t_{\emph{expire}}$\xspace}
\newcommand\ddraw{$d_{\emph{draw}}$\xspace}

\newcommand\dredeem{$d_{\emph{redeem}}$\xspace}
\newcommand\tktrate{$tkt_{\emph{rate}}$\xspace}
\newcommand\lesc{$l_{\emph{esc}}$\xspace}
\newcommand\drawlen{$\emph{draw}_{\emph{len}}$\xspace}

\newcommand{\ignore}[1]{}

\pagestyle{plain}

\begin{document}
%
% paper title
\title{\Large \bf\sysname: Practical Concurrent Processing of Micropayments}

\author{Ghada Almashaqbeh\inst{1}\thanks{Most work done while at Columbia supported by NSF CCF-1423306.} \and
Allison Bishop\inst{2,3}\thanks{Supported by NSF CCF-1423306 and NSF CNS-1552932.} \and
Justin Cappos\inst{4}}
\authorrunning{G. Almashaqbeh et al.}
% First names are abbreviated in the running head.
% If there are more than two authors, 'et al.' is used.
%
\institute{CacheCash Development Company, NY, USA
\email{ghada@cachecash.com}\\ 
\and 
Columbia University, NY, USA
\email{allison@cs.columbia.edu}\\ 
\and 
Proof Trading, NY, USA\\
\and
New York University, NY, USA
\email{jcappos@nyu.edu}}

\maketitle

\begin{abstract}
Micropayments are increasingly being adopted by a large number of applications. However, 
processing micropayments      
individually can be expensive, with transaction fees exceeding  
the payment value itself. By aggregating these small transactions into a few larger ones, 
and using cryptocurrencies, today's decentralized probabilistic micropayment schemes can 
reduce these fees. Unfortunately, existing solutions force micropayments to be   
issued sequentially, thus to support fast issuance rates a customer needs to 
create a large number of escrows, which bloats the blockchain. 
Moreover, these schemes incur a large 
computation and bandwidth overhead, which limit their applicability in large-scale systems.

In this paper, we propose \sysname, the first decentralized  
probabilistic framework that supports  
concurrent micropayments. \sysname introduces a novel 
escrow setup that enables a customer to concurrently issue payment tickets  
at a fast rate using a \emph{single} escrow. \sysname is also cost effective   
because it allows for ticket exchange 
using only one round of communication, and it aggregates the micropayments 
using a lottery protocol that requires only secure hashing. 
Our experiments show that \sysname can process thousands of tickets   
per second, which is around 1.7-4.2x times the rate 
of a state-of-the-art sequential micropayment system. Moreover, \sysname 
supports any ticket issue rate over any period using only 
one escrow, while the sequential scheme 
would need more than 1000 escrows per second to permit high rates. 
This enables our system to further reduce  
transaction fees and data on the blockchain by $\sim50\%$.
\end{abstract}

\section{Introduction}
\label{intro}
Micropayments, or payments in pennies, are increasingly being used 
by applications as diverse as 
ad-free web surfing, online gaming, and peer-assisted 
service networks~\cite{Pass15}. This paradigm allows participants 
to exchange monetary incentives at a small 
scale, e.g., pay per minute in online games. Such a fine-grained   
payment process has several 
advantages, including a great deal of flexibility for customers who may 
stop a service at any time. In addition, it reduces the financial risks 
between mutually-distrusted participants, where there is no 
guarantee that a client will pay after being served, or  
that a server will deliver service when paid in advance.

However, processing these small payments individually can incur 
high transaction fees that exceed the few pennies received. 
For example, the average base cost of a debit or credit card 
transaction in the US is around  
21 to 24 cents, and 23 to 42 cents~\cite{debit-fee,credit-fee}, 
respectively. In cryptocurrencies such a fee could be even higher, e.g., above $\$$1 
in Bitcoin~\cite{bitcoin-fee}. 
Beside this financial drawback, handling micropayments 
individually can impose a huge workload on the system, 
and may explode the  
log needed for accountability purposes. Thus, there is a need 
for a payment aggregation mechanism that records fewer transactions 
with values that still compensate properly for the small payments 
received to date.

Probabilistic micropayment schemes have 
emerged as a solution that fits the criteria outlined 
above~\cite{Wheeler96,Rivest97,Micali02,Rivest04}. In these models, 
the amount of required payments is locked in an escrow and 
micropayments are issued as lottery tickets. Each ticket has a probability $p$ 
of winning a lottery, and when it wins, produces a transaction of $\beta$ currency 
units. This means that, on average, only one large transaction is processed 
out of a batch of $\frac{1}{p}$ tickets. Unfortunately, these   
early proposals rely on a trusted party
to audit the lottery and manage payments. Such a centralized approach 
may increase the 
deployment cost and limit the use of the payment service 
to systems of fully authenticated participants~\cite{Chiesa17}.

As cryptocurrencies evolved, a number of initiatives 
have attempted to convert these schemes 
to distributed ones~\cite{Pass15,Chiesa17}. This is done by replacing 
the trusted party with the miners, and utilizing the blockchain to 
provide public verifiability of system operation. Yet, these 
approaches have several drawbacks that may   
hinder their usage in large-scale systems.  First, they force a  
customer to issue micropayments sequentially using 
the same escrow. This means a new ticket cannot be issued 
until it is confirmed that the previous one  
did not win, which requires a merchant to report the lottery outcome 
back to the customer. To issue tickets at a fast rate under this structure, this customer needs to 
create a large number of escrows, which increases the amount of data 
on the blockchain. Second, these schemes rely on 
computationally-heavy cryptographic primitives~\cite{Pass15,Chiesa17}, and  
several rounds of communication to exchange payments~\cite{Chiesa17}, which 
incur a large overhead. 
Such performance issues reduce the potential benefits of micropayments.

This paper proposes a solution to these drawbacks by introducing \sysname, the 
first decentralized probabilistic framework that supports  
\emph{concurrent} micropayments. \sysname features a novel payment 
setup that allows a customer to issue micropayments 
in parallel and at a fast rate using a \emph{single} escrow that can pay many
winning tickets.  
This is achieved by having the customer specify the total number of tickets it 
may issue, and provide an escrow balance that covers all  
winning tickets under its payment setup.

\sysname is also cost effective because it introduces a lightweight non-interactive 
lottery protocol. This protocol requires only secure hashing and allows a payment exchange using   
only one round of communication without demanding the merchant to 
report anything to the customer. Furthermore, this protocol is the first to eliminate situations where 
all lottery tickets may win or lose the lottery. 
Although the probability of these events is very small, it may discourage 
customers from using the system since such a possibility, i.e., to pay much 
more than the expected payments, may impose a strong psychological 
obstacle~\cite{Micali02}. Moreover, accounting for the worst 
case when all, or almost all, tickets win requires 
a large escrow balance, which increases the collateral cost. Our 
protocol solves this problem by selecting  an \emph{exact} number of winning 
tickets each round (where a round is the time needed to mine a block on the 
blockchain). All tickets issued in 
the same round are tied to a \emph{lottery draw} value in a future block 
on the blockchain, which is used to determine the set of winning tickets through  
an iterative hashing process. The security of this protocol, and the 
whole system, is enforced using both cryptographic and financial techniques. The latter 
requires a customer to create a penalty escrow that is revoked upon cheating, with a 
lower bound derived using a game theoretic modeling of the system.

To evaluate its efficiency, we experimentally test \sysname's performance  
and compare it to MICROPAY~\cite{Pass15}, a 
state-of-the-art sequential micropayment scheme. Our results 
show that a modest merchant machine in \sysname is able to process  
2,240 - 10,500 ticket/sec, which is 
around 1.7-4.2x times the rate in MICROPAY, with $60\%$ reduction in 
the aggregated payment size. Furthermore, a modest customer machine in \sysname is 
able to concurrently issue more than 33,000 ticket/sec 
using \emph{one} escrow over any period, while MICROPAY requires the creation of 
more than 1000 escrows per second to 
support a comparable issue rate. This allows \sysname 
to reduce transaction fees and amount of data on the blockchain in a video 
delivery and online gaming applications by $\sim50\%$.

\section{Related Work}
\label{related-work}
To orient readers to the current state-of-the-art in  
probabilistic micropayments, in this section we review 
prior work done in this area. In addition, we review an alternative payment 
aggregation mechanism, called payment networks~\cite{Decker15,Poon15}, 
focusing on its limitations when used to handle micropayments. \\

\noindent{\bf Probabilistic Micropayments.}
The idea of probabilistic micropayments dates 
back to the seminal works of Wheeler~\cite{Wheeler96} and 
Rivest~\cite{Rivest97}. In these schemes, a customer and a merchant run  
the lottery on each ticket by using a simple coin tossing 
protocol. Thus, there is a chance than more, or less, tickets than 
expected may win. All of these schemes rely on a centralized 
bank to track and authorize payments. 
This imposes additional overhead on the users  
who have to establish business relationships 
with this bank. Also, it limits the use 
of the service to only fully authenticated users. Therefore, 
this centralization issue is viewed as the 
main reason for the limited adoption of such solutions~\cite{Chiesa17}.

Systems using cryptocurrency-based probabilistic micropayments 
have the potential to overcome both the cost and efficiency problems 
inherent in earlier schemes. To the best 
of our knowledge, only two such schemes have been proposed 
to date in the literature, MICROPAY~\cite{Pass15} and DAM~\cite{Chiesa17}.

MICROPAY translates the scheme of Rivest~\cite{Rivest97} into an 
implementation on top of a cryptocurrency system. 
Instead of using an authorized bank account, customers create escrows on 
the blockchain that they use to issue lottery tickets. For the 
lottery protocol, MICROPAY implements a similar interactive 
coin tossing protocol, and adds an alternative non-interactive 
version that reduces the communication complexity (a merchant still has 
to report the lottery result back to the customer). 
However, the latter is computationally-heavy since it requires public key  
cryptography-based operations and a non-interactive zero 
knowledge (NIZK) proof system. Moreover, MICROPAY only supports sequential 
micropayments as mentioned earlier. DAM shares similar  
constraints, but it adds the feature of preserving user privacy (not like 
MICROPAY that is public), where it extends Zerocash~\cite{Sasson14} 
primitives to implement anonymous micropayments.

We believe that the added blockchain transactions due to sequential payments, 
coupled with the high computation cost, point to the 
need for optimized approaches 
that support concurrent micropayments at a lower 
overhead. This need is the motivation behind building \sysname.\\

\vspace{-8pt}
\noindent{\bf Payment Channels and Networks.}
This payment paradigm was originally developed to handle 
micropayments in Bitcoin~\cite{bitcoinj}, where it relies on 
a similar concept of processing most of the small payments 
locally. Later on, it was geared to enhancing the scalability of 
cryptocurrencies~\cite{Hearn12,Decker15,Poon15,Miller17,Green16,Malavolta17}, 
where utilizing off-chain processing to reduce on-chain traffic helps increase the 
transaction throughput at a lower overhead.

A payment channel is a contract 
between a customer and a merchant tied to a shared escrow fund. 
The ownership of this fund is adjusted over time based on the 
off-chain transactions, or local payments,   
made to date. Only two transactions 
are logged on the blockchain per channel, the opening transaction and the 
closing one that expresses   
the latest state of the fund ownership.

In general, payment channels and networks suffer from the high collateral cost 
of setting up multiple escrows when constructing payment paths between 
transacting parties. 
These costs may indirectly push the network towards 
centralization~\cite{ct-article}. This is because only wealthy parties can 
afford multiple escrows to establish payment 
channels, and hence, most users will rely on these parties, 
or hubs, to relay the off-chain transactions. In addition, each hub on the 
path charges a fee to relay payments. With micropayments, such a setup 
would be infeasible because 
these fees could be much larger than the payments themselves. Probabilistic 
approaches, on the other hand, are more flexible 
in allowing several parties to be paid using the same escrow. 
And by doing so, they reduce the collateral cost and 
eliminate any fees when exchanging lottery tickets. 
Hence, distributed probabilistic micropayments provide a better 
solution for handling small payments in cryptocurrency systems.

\section{Threat Model}
\label{threat-model}
The reliance on  
off-chain transactions in distributed probabilistic micropayments creates   
the potential for various types of attacks. In this section, we outline 
a threat model that accounts for these attacks, which 
guided the design of \sysname. In developing this model, we make the following assumptions:
\begin{itemize}
\itemsep0em
\vspace{-3pt}
\item No trusted party exists. 

\item Participants are rational, meaning that they may follow the protocol 
without violation, or deviate from it, based on what will maximize their 
utility gain.

\item The underlying cryptocurrency scheme is secure in the sense 
that the majority of the mining power is honest. This means that  
the confirmed state of the blockchain contains only valid transactions, 
and that an attacker who tries to mutate or fork the blockchain will 
fail with overwhelming probability.

\item Hash functions are modeled as random oracles, and 
the hash values of the blocks on the blockchain are 
modeled as a uniform distribution.

\item Efficient adversaries cannot break the basic cryptographic 
building blocks (SHA256, digital signatures, etc.) with non-negligible 
probability. 

\item Communication between customers and merchants 
takes place over a channel that provides integrity, confidentiality, and 
authenticity, such as TLS/SSL.
\end{itemize}

We used ABC~\cite{Almashaqbeh19} to build a threat model for 
MicroCash (a detailed version of this threat 
model is available online~\cite{material}). During this process, we identified 
the assets to be protected in distributed probabilistic micropayments, which 
include the escrows, the 
lottery tickets, and the lottery protocol. Then, by analyzing the 
security requirements of these assets, we 
produced the broad threat categories in such systems. Our list includes the 
following:
\begin{itemize}
\itemsep0em
\vspace{-3pt}
\item {\bf Escrow overdraft:} A customer creates a 
payment escrow insufficient for honoring the winning lottery tickets,  
or creates a penalty deposit that does not cover the cheating punishment imposed 
by the miners. Such a threat could be the result of creating 
small balance escrows, or front running attacks 
in which a customer withdraws the escrows before paying.

\item {\bf Duplicate ticket issuance:} A customer issues 
lottery tickets with the same sequence number to several merchants. 
This leads to issuing more tickets than the escrow can 
cover, allowing the customer to obtain more 
service than it pays for. 

\item {\bf Invalid payments:} A malicious customer hands merchants  
lottery tickets that do not comply with its payment setup or with the 
system specifications. Because 
these tickets will be rejected by the miners if they win the lottery, the 
customer can avoid paying merchants.

\item {\bf Unused-escrow withholding:} An attacker prevents or delays 
a customer from withdrawing its unused escrows. For example, merchants may 
delay claiming their winning lottery tickets to keep the payment escrow on hold.

\item {\bf Lottery manipulation:} An attacker attempts to influence the outcome 
of the lottery draw, and hence, bias the payment process. 

\item {\bf Denial of service (DoS):} This is a large threat category that 
threatens any distributed system. This work focuses on attacks related 
to the payment process. For example, an attacker may 
monitor the network and drop all lottery tickets to 
disturb the service.  
\end{itemize}

Note that dealing with malicious merchants who collect lottery tickets and 
do not deliver a service is outside the scope of \sysname. The same is true for dealing 
with malicious customers who may obtain the service without paying. In this 
work, we are concerned 
with the payment scheme design, rather than how to exchange service 
for a payment, which is part of an application design.

In addition, \sysname does not address 
payment anonymity. Addressing this issue securely, 
while preserving the low overhead of \sysname, is a direction of our future 
work.

\section{\sysname Design}
\label{design}
Having outlined the security threats to probabilistic micropayments, and 
the limitations of existing solutions,  
this section presents the design of \sysname, a concurrent micropayment system
that addresses these issues. We start with an overview of the system, 
followed by a more detailed technical description.

\begin{figure}[h!]
\centerline{
\includegraphics[height= 1.4in, width = 0.55\columnwidth]{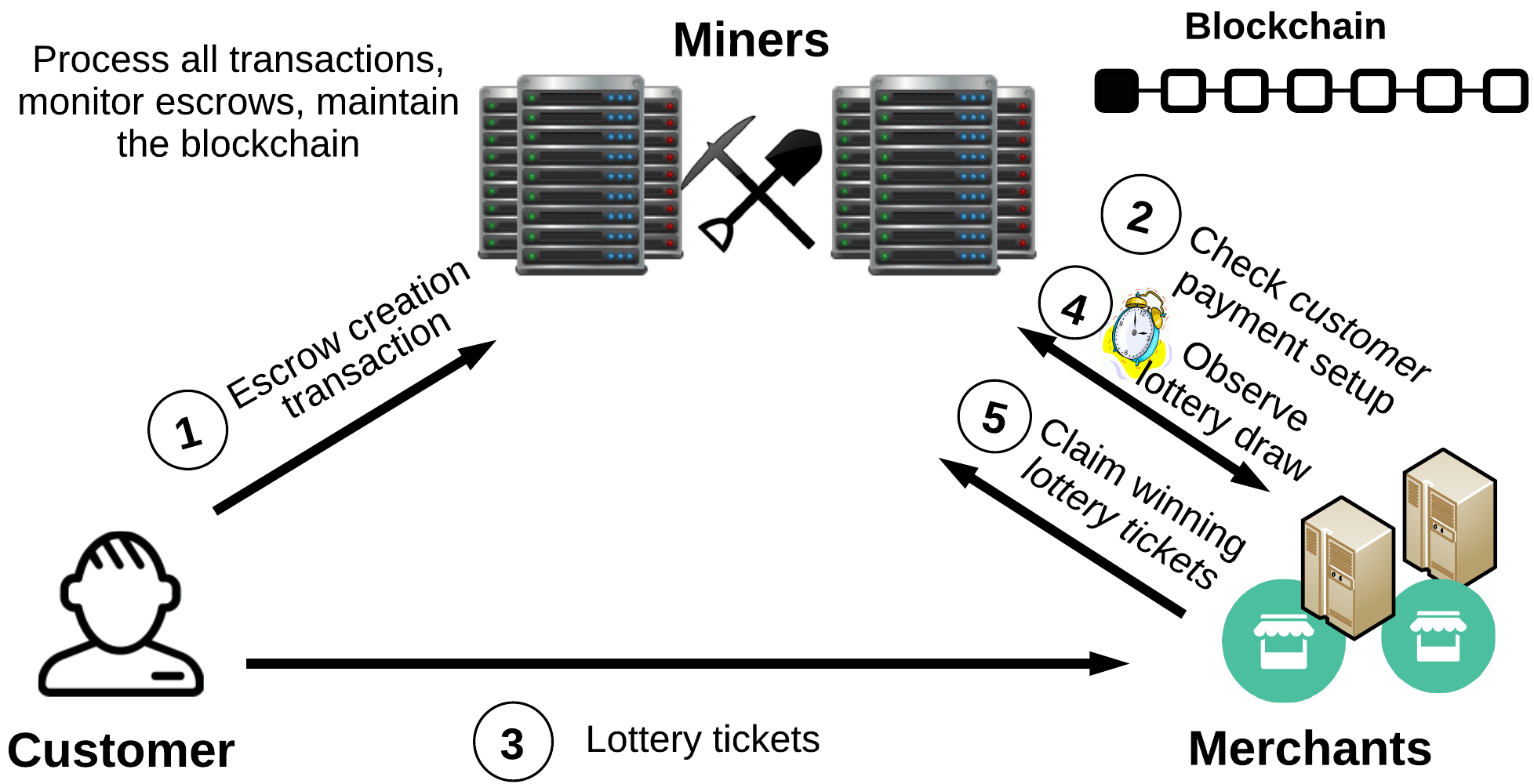}}
\vspace{-7pt}
\caption{Flow of operations in \sysname.} \label{microcash-flow}
\end{figure}

A high level illustration of \sysname, that also captures the remainder of this 
section's organization, is found in Figure~\ref{microcash-flow}. 
As shown, during the payment setup ({\bf Step 1}, Section~\ref{escrow-setup}), each 
customer issues a transaction creating two 
escrows: payment and penalty. The customer uses the former   
to make payments in the form of lottery tickets, while the 
miners use the latter to financially punish this customer if it cheats. 
Payment exchange starts as soon as  the  
escrow transaction is confirmed on the blockchain. At that time, 
merchants can check the escrow setup before 
transacting with the customer ({\bf Step 2}). In exchange for the delivered 
service, the customer issues these merchants lottery tickets 
according to a ticket issue schedule that limits the ticket issue rate  
over a set period ({\bf Step 3}, Section~\ref{pay-tkt}). To redeem these payments, a merchant 
keeps each of its tickets until the lottery draw time of this ticket. 
It then observes a value derived based on the block 
mined at that time to determine if this ticket won 
({\bf Step 4}, Section~\ref{lottery-protocol}). If it is a winning ticket, the merchant 
can claim currency 
from the customer's escrow during the ticket redemption period 
({\bf Step 5}, Section~\ref{claim-tkt}). 
This interaction continues until the end of the escrow lifetime. 
At that time, and when all issued tickets expire, the customer can spend any 
remaining funds.

\subsection{Escrow Setup}
\label{escrow-setup}
\sysname introduces a novel escrow setup that allows multiple winning tickets
to be redeemed, which enables both concurrent ticket issuance and reduces
the amount of on-chain escrow-related data. This setup provides techniques 
to determine the needed escrow balance to cover all concurrent 
tickets, and to track the issuance of these tickets 
in a distributed way.

\subsubsection{Escrow creation.}
As an off-chain payment scheme, 
\sysname must ensure that customers can and will pay.  
This includes honoring winning tickets, and, if caught cheating,  
complying with a stipulated financial punishment. 
To satisfy these requirements, each customer must create 
payment and penalty escrows with sufficient funds to cover 
both eventualities.

Given that each payment escrow must be tied to a penalty escrow,  
a customer sets up both using one creation transaction. 
This transaction provides funds to be locked under each 
escrow balance, where we refer to the payment and penalty escrow 
balances as \besc and \bpen, respectively. It also configures a set of 
parameters that influence how the value 
of both \besc and \bpen are computed, and 
how they are to be spent. These parameters, whose values are 
specified by the customer possibly after negotiating with the 
merchants, include the following:
\begin{itemize}
\itemsep0em
\item The lottery winning probability $p$.

\item The currency value of a winning lottery ticket $\beta$.

\item The ticket issue rate \tktrate, which is the maximum number  
of tickets a customer is allowed to hand out per round. This is used
to calculate which \emph{ticket sequence numbers} are valid within each ticket 
issuing round.

\item A lottery draw round length, denoted as \drawlen, such that  
$draw_{\emph{len}} \in \{1, \dots, c\}$ for some small system parameter 
$c$. The customer 
has to configure \drawlen, $p$, and \tktrate in a way that makes  
$p$\tktrate\drawlen of an integer value (this is the number 
of winning tickets in a lottery draw). 

\item The set of beneficiary merchants that can be paid using the escrow, 
where the size of this set is denoted as $m$.
\end{itemize}

Computing \besc and \bpen based on the above parameters 
proceeds as follows. To permit concurrent micropayments, \besc must be large 
enough to pay all winning tickets tied to an escrow. Given that each winning 
ticket has a value of $\beta$ currency units, and that there are $p$\tktrate\drawlen  
winning tickets per \drawlen rounds, \besc can be simply computed as follows (where 
\lesc is the escrow lifetime in rounds, and there are \lesc/\drawlen lottery draws):
\begin{equation}\label{payment-balance}
B_{\emph{escrow}} = \beta p tkt_{\emph{rate}} l_{\emph{esc}}
\end{equation}

For \bpen, we compute a lower 
bound for this deposit using an economic analysis that accounts for 
the additional utility gain a customer may obtain by cheating. This 
bound is given by the following equation:
\begin{footnotesize}
\begin{equation}\label{penalty}
B_{\emph{penalty}} > (m-1)p\beta tkt_{\emph{rate}} \emph{draw}_{\emph{len}} \bigg(\frac{1-p}{1-\rho^{-1}} + \emph{draw}_{\emph{len}}\Big((1-p)(d_{\emph{draw}} -1)+ d_{\emph{redeem}}\Big)\bigg)
\end{equation}
\end{footnotesize}

\noindent where \ddraw 
is the lottery draw period in rounds, \dredeem is the ticket  
redemption period in rounds, and $\rho = {a \choose b}$ 
such that $a = {tkt_{\emph{rate}} draw_{\emph{len}}}$ and $b = (1-p)a$ (more 
about these parameters 
in Section~\ref{lottery-protocol}). This lower bound 
ensures that the financial punishment  
exceeds the additional utility gain of cheating, and hence, makes cheating 
unprofitable compared to acting honestly. The 
full details of deriving this bound 
are found in Section~\ref{penalty-analysis}.

Verifying the correctness of a payment setup is performed by 
the miners upon receiving the escrow creation transaction. First, they verify 
that the customer owns the input funds.  
Then, the miners use \besc to compute the escrow lifetime as \lesc = 
$\frac{B_{\emph{escrow}}}{\beta p tkt_{\emph{rate}}}$. After that, 
they check that both \lesc and $p$\tktrate\drawlen  
are of integer values, \drawlen is within the allowed range, and that \lesc is  
multiples of \drawlen. Lastly, the 
miners verify that \bpen satisfies the bound given above. If all these 
checks pass, the miners add the escrow transaction to the blockchain. Otherwise,  
they reject the escrow by dropping its transaction.

\subsubsection{Escrow management.}
In \sysname, the escrow funds can be spent only for 
a restricted set of transactions. This set includes claiming winning tickets, 
presenting proofs-of-cheating, and (after the escrow lifetime is over) enabling a 
customer to spend its unused escrow funds.

To track the locked funds, miners maintain a state for each escrow 
in the system. This state includes the following:
\begin{itemize}
\itemsep0em
\item The ID of the escrow, which is a random value generated 
by the miner who adds the escrow creation transaction to the 
blockchain.

\item The balances of both the payment and penalty escrows.

\item The public key of the owner customer, which 
is used to verify all signed tickets 
that are issued using this escrow.

\item The values of $p$, $\beta$, \lesc, \tktrate, \drawlen, 
and the set of beneficiary merchants (both the public key of each 
merchant and a corresponding index).

\item An escrow refund time, denoted as \trefund, at which the 
customer can spend any remaining funds. Miners set this time to be equal to 
the expiry time of the 
tickets issued in the last round of an escrow lifetime.
\end{itemize}

Ticket issuance using an escrow must follow a schedule based upon the 
tickets’ sequence numbers. That is, if an escrow supports a rate of  
\tktrate tickets per round, then in the first round tickets with sequence 
numbers $0$ to \tktrate-1 may be issued. Then, in the second round 
tickets with sequence number range \tktrate to 
$2$\tktrate-1 can be issued, and so on until the last 
round of an escrow lifetime. Merchants will accept tickets in the current round 
with sequence numbers that follow this assignment schedule. 
In order to deal with the fact that customers and merchants may have 
inconsistent views of the blockchain, 
and hence, may not agree about what the current
round is (i.e., current height of the blockchain), merchants will also 
accept tickets from the prior and next round given that these tickets use the 
correct sequence number range.

An example of a ticket issuing schedule is found in Figure~\ref{ticket-issue-schedule}. 
As shown, the escrow creation 
transaction is published at round 10 and confirmed at round 16. This escrow 
has \lesc = 3 rounds, and allows a ticket issue rate of 1000 tickets per round. 
Thus, the customer has 3 ticket issuing rounds, starting at round 17, with 
the sequence number ranges shown in the figure.

\begin{figure}[t!]
%\vspace{-10pt} 
\centerline{
\includegraphics[height= 1.3in, width = 0.8\columnwidth]{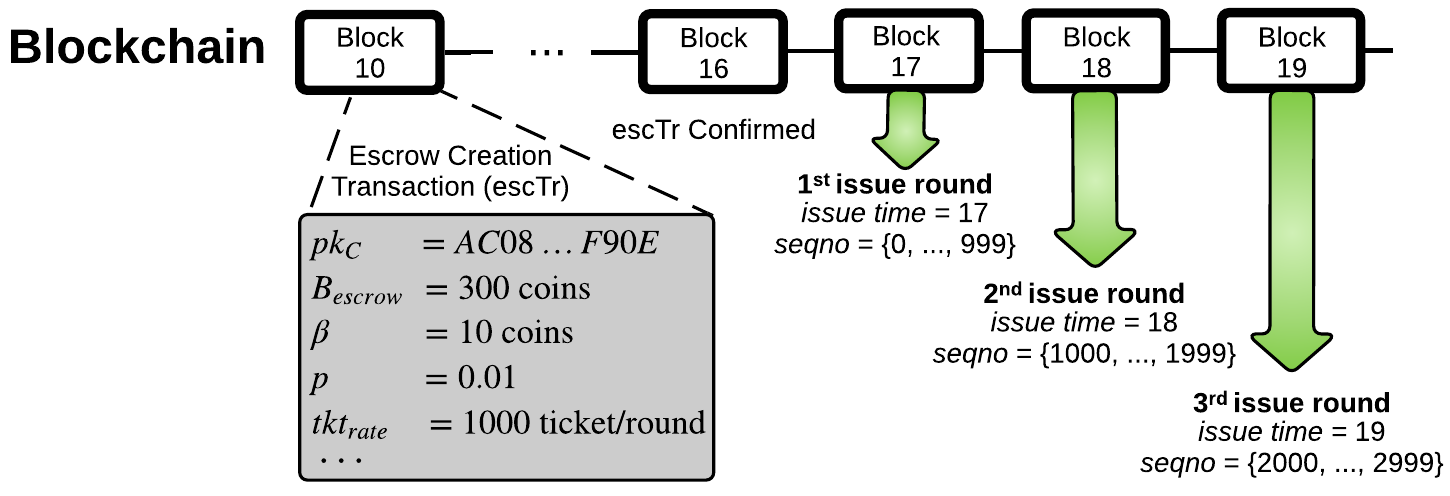}}
\vspace{-8pt}
\caption{An example of a ticket issuing schedule.}
\label{ticket-issue-schedule}
\end{figure}

The miners update the escrow state based on the 
escrow related transactions (mentioned earlier) they process.  
For example, redeeming a winning ticket reduces 
\besc by $\beta$ coins, and receiving a valid proof-of-cheating against the 
customer causes the miners to burn the funds in 
\bpen. All these transactions are logged on the 
blockchain, which permits anyone to validate the state.

The miners discard an escrow state once all tickets tied to this 
escrow expire, which happens at time \trefund, or when an escrow is broken 
after receiving a valid proof-of-cheating (proof-of-cheating 
is discussed in Section~\ref{proof-of-cheating}). At that time, the customer may 
spend the remaining funds of its payment escrow 
(if any) and its penalty deposit (if not revoked).

\subsection{Paying with Lottery Tickets}
\label{pay-tkt}
After the escrow is confirmed on the blockchain, a customer 
can start paying for service by giving merchants lottery tickets. A lottery 
ticket  $tkt_L$ is a structure containing several fields as follows: 
\begin{equation}
tkt_L = index_M||id_{\emph{esc}}||seqno||\sigma_C
\vspace{-5pt}
\end{equation}

\noindent where $index_M$ is the recipient merchant index 
as listed in the escrow state, $id_{\emph{esc}}$ is the 
escrow ID, $seqno$ is the ticket sequence 
number, and $\sigma_C$ is the customer's signature. The 
$seqno$ field, along with $id_{\emph{esc}}$, identifies a ticket, which 
also provides means for ticket tracking in the system.
Note there is no need to include any information about the escrow setup or 
the parties' public keys in the ticket 
itself. Merchants and miners can look these up on the blockchain using $id_{\emph{esc}}$.

When issuing a ticket, the customer fills in the above fields and signs the
ticket using the secret key tied to the public key the customer used 
when creating the escrow. The 
ticket $seqno$ can be any sequence number within the range assigned 
to the current ticket issue round. The 
customer can continue issuing lottery tickets, without waiting the lottery results  
of previously issued ones, until it 
finishes all sequence numbers in this range. After that, it must 
wait the next round to generate more tickets.

Upon receiving a ticket, a merchant verifies it as follows: 
\begin{itemize}
\itemsep0em
\item Check that the escrow is not broken
\item Check that its index $index_M$, that appears in the ticket, is 
identical to the one listed in the escrow state. 
\item Verify that $seqno$ is within the valid range based on the 
ticket issuing schedule. (As mentioned before, to handle inconsistencies 
in the blockchain view, tickets
from the previous or next issuance round can be accepted.)
\item Verify $\sigma_C$ over the ticket using the customer's public key.
\end{itemize}

If any of the above checks fails (except the fourth one), the merchant 
drops the ticket. On the other hand, if the ticket has 
an out-of-range sequence number (i.e., larger than the maximum sequence 
number allowed by the escrow), the recipient merchant can issue a proof-of-cheating that 
will cost the customer its penalty deposit. Otherwise, if all the above 
checks pass, the merchant accepts the ticket and keeps it until its 
lottery draw time.

\subsection{The Lottery Protocol}
\label{lottery-protocol}
\sysname introduces a lightweight  
lottery protocol that relies solely on secure hashing. This protocol 
does not require any interaction between the customer and the 
merchant. Instead, it utilizes only the state of the blockchain, where 
the lottery draw outcome is determined by a value derived from the block 
mined at the lottery draw time.

To specify the lottery draw time, \sysname defines 
two system parameters, \ddraw and \drawlen mentioned earlier. \ddraw   
represents the least number of rounds a ticket has to wait after its issue 
round (which we call \tissue) until it   
enters the lottery. \drawlen determines the number of consecutive 
ticket issuing rounds that will have all their lottery tickets enter the same 
lottery draw\footnote{Since \drawlen affects \tdraw of a ticket, \sysname 
specifies a small interval for its possible 
values to prevent a customer from excessively delaying paying merchants.}. 
Hence, if \drawlen = 1, then the draw time 
\tdraw of a ticket is computed as \tdraw = \tissue + \ddraw. On the other hand, 
if \drawlen $> 1$, then \tdraw of a ticket is \tdraw of the last ticket 
issuing round in the contiguous set of rounds.

A clarifying example of determining the lottery draw time is found in 
Figure~\ref{lottery-draw-time}. As shown, starting with the first ticket issuing 
round, which is 28 in the figure, each set of contiguous \drawlen rounds enter 
the lottery together. For example, all tickets issued in rounds 28, 29, and 30 enter 
the lottery at round 40, which is 10 rounds after the last ticket issue round in this set.

\begin{figure}[t!]
\centerline{
\includegraphics[height= 0.6in, width = 0.8\columnwidth]{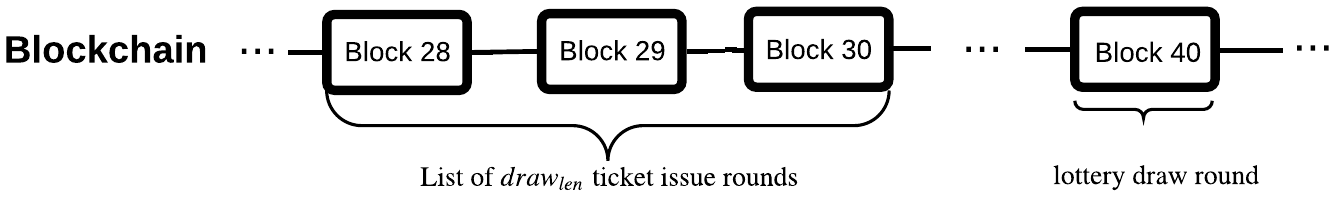}}
\vspace{-8pt}
\caption{An example of lottery draw time (\drawlen = 3, $p = 1/300$, \tktrate = 1000, and \ddraw = 10). } 
\label{lottery-draw-time}
\end{figure}

Accordingly, whether a ticket wins or loses depends on 
a lottery draw value tied to the block mined at time \tdraw. 
This value is computed using a simple verifiable delay function 
(VDF)~\cite{Boneh18} that is evaluated over this block. This evaluation takes 
a period of time, hence the name delay function, where this period is a 
system parameter. Consequently, when a miner mines the block at index \tdraw, 
it cannot tell immediately which ticket will win or lose. This miner first has to 
compute the VDF over this block.

We instantiate this VDF using iterative hashing,   
where the number of iterations is set to a value that delays 
producing the output 
by the period specified in the system. In addition, we let the miners compute this 
function as part of the mining process. 
That is, when a miner mines a new block, it evaluates the VDF   
over the previous block. Therefore, 
the VDF value of the block at index \tdraw appears 
on the blockchain when the block at index \tdraw + 1 is mined.

Accordingly, in our protocol a merchant keeps a ticket $tkt_L$ until its  
lottery draw time \tdraw. Then, after observing the VDF value of the block 
mined at that time, the miners, and any party in the system, can compute the 
set of winning sequence numbers for that round as 
follows. First, the hash of the VDF value along with the escrow ID is computed, which we 
call $h_1$, and then $h_1$ is mapped to a sequence number within the assigned range 
of the ticket issuing rounds tied to \tdraw. To obtain the second winning ticket 
sequence number, the hash of 
$h_1$ is computed to obtain $h_2$, and then $h_2$ is mapped to a ticket sequence 
number in the given range. If a collision occurs, i.e., a previously seen sequence 
number is produced, 
it is discarded and the process proceeds with hashing $h_2$ to obtain $h_3$, and so on. 
This continues until a set of distinct 
$p$\tktrate\drawlen winning sequence numbers is drawn\footnote{We design a  
version of this lottery protocol with independent ticket winning events 
in Appendix~\ref{lottery-protocol-ind}. This version can be used in case it is infeasible 
in some applications to configure $p$\tktrate\drawlen to be an integer.}.

The previous process is clarified by the example depicted in 
Figure~\ref{lottery-example}, which has the same setup as in 
Figure~\ref{lottery-draw-time}. As shown, and assuming that round 28 is the first 
ticket issuing round, the ticket has been issued at round 29, and hence, 
it entered the lottery at round 40. The VDF value of the block with index 
40 appears inside block 41. By using this value, a set of winning sequence numbers 
is chosen, based on which the ticket in the figure is a winning one because its 
sequence number is within this set.

\begin{figure}[t!]
\centerline{
\includegraphics[height= 1.7in, width = 0.9\columnwidth]{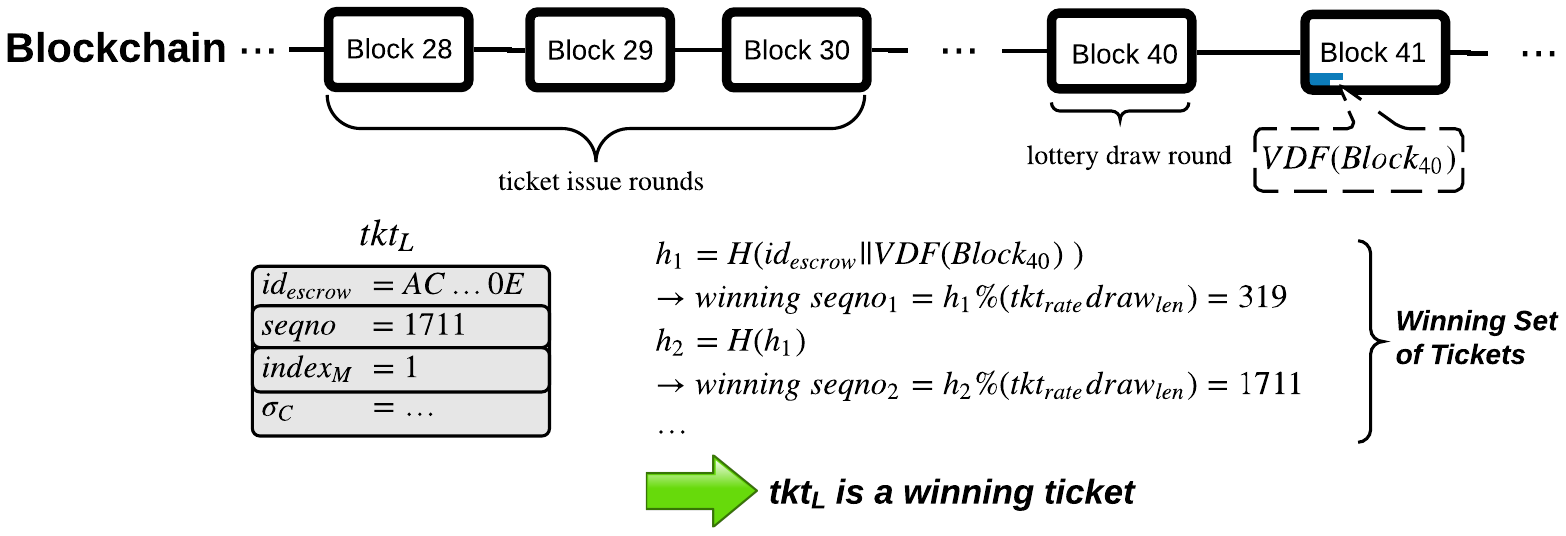}}
\vspace{-7pt}
\caption{Lottery draw example (same setup as in Figure~\ref{lottery-draw-time}). } 
\label{lottery-example}
\end{figure}

Note that the lottery draw involves only 
values that are part of the escrow state. In other words, it relies on  
parameters that the issuing customer cannot manipulate, which do not 
include the merchant recipient address. This  
means that a ticket's chance of winning the lottery is not affected by 
who owns it. In addition, this means that if a customer issues  
tickets with duplicated sequence numbers to multiple merchants, all these 
tickets will win or lose together. If the 
tickets win, detecting cheating is trivial because both merchants will publish
their winning tickets to the blockchain to redeem the tickets.

\subsection{Claiming Winning Tickets}
\label{claim-tkt}
After the lottery draw, a merchant can collect currency from the 
customer's escrow by redeeming its winning tickets (if any). This is 
done by issuing a redeem transaction that has the winning ticket 
as input, and has $\beta$ coins 
directed to the merchant's address as output.

To allow the miners to resolve tickets and release escrow funds back to the 
customer in a reasonable timeframe, \sysname specifies a redeem period for 
each ticket. This is done by defining a system parameter called \dredeem  
that determines the number of rounds 
during which a ticket can be redeemed. After this period, a ticket 
expires, which happens at time 
\texpire = \tdraw + \dredeem. Thus, \dredeem must be set to  
a value that allows merchants to redeem their winning tickets.

After receiving a redeem transaction, the miners process it as follows: 
\begin{itemize}
\itemsep0em
\item Check that the format of the transaction complies with the system 
specifications.

\item Verify the redeemed ticket as outlined in Section~\ref{pay-tkt}. 

\item Verify that the ticket is a winning one by checking that its sequence 
number is among the winning set tied to time \tdraw of this ticket.

\item Check that the ticket is not expired. 

\item Verify the merchant's signature over the redeem transaction using the 
public key corresponding to $index_M$ found in the escrow state. 
This is needed to 
prevent participants from redeeming tickets they do not own.

\item Check that no other ticket with the same sequence number and   
tied to the same escrow has already been redeemed. If it is, this is
a proof of duplicate ticket issuance and is used as
a proof-of-cheating against the customer. 
\end{itemize}

If all these checks pass, miners approve the redeem transaction 
and update the escrow state accordingly. Otherwise, they drop 
an invalid transaction and, if a proof-of-cheating is produced, 
revoke the customer's penalty deposit.

\subsection{Processing Proof-of-cheating}
\label{proof-of-cheating}
A proof-of-cheating is a special transaction that 
can be presented to the miners by any party who witnesses 
a cheating incident. In \sysname, such an incident could be issuing more tickets than what 
an escrow can cover, i.e., exceed the maximum $seqno$ an escrow may allow, 
or duplicate ticket issuance. A signed lottery ticket with an out-of-range sequence 
number or signed 
tickets with duplicated sequence numbers are publicly verifiable proofs against 
the issuing customer.

Upon verifying cheating, miners punish the   
customer by revoking the penalty escrow tied to  
its payment escrow referenced in the ticket as follows. In case of duplicate 
ticket issuance, the miners first pay all duplicated winning tickets 
from the payment escrow, if it is sufficient, and from the penalty deposit 
thereafter. Then, they publish an 
escrow break transaction containing the proof-of-cheating on the 
blockchain. This transaction burns the revoked penalty 
deposit rather than providing them to another party. This is done to 
mitigate the case that if the funds
are provided to another party, that party may have colluded with the
customer to receive those funds. Respecting the lower bound of \bpen, as specified 
by equation~\ref{penalty}, ensures that all the 
aforementioned cheating behaviors are less profitable than acting 
in an honest way.  
Hence, it makes such behaviors unappealing to rational customers.

\section{Computing a Lower Bound for \bpen}
\label{penalty-analysis}
\vspace{-4pt}
In this section, we compute a lower bound for the penalty  
deposit \bpen required to deter cheating. This is done using a 
game theoretic approach that quantifies the additional utility 
gain, or monetary profit, a malicious customer could accrue as 
compared to an honest one. By setting the penalty deposit to 
at least equal this additional utility, cheating becomes 
less profitable in expectation than acting honestly, and hence, 
becomes less unappealing to rational customers.

In what follows, we present this analysis including the 
malicious strategies addressed, the game setup, and a definition for the utility 
gain function. Finally, we state and prove a lower bound for \bpen. \\

\vspace{-4pt}
\noindent{\bf Covered malicious strategies.} 
In \sysname, a penalty escrow is revoked upon the detection of 
two types of malicious events: issuing duplicated tickets or issuing invalid 
payments. The utility gain of any of these malicious strategies 
depends on the length of the cheating detection period, i.e., the time needed 
to detect a cheating incident. Throughout this period, the cheating customer is 
still perceived as honest by merchants, and so can continue 
cheat and increase its utility gain. Consequently, the 
longer the detection period lasts, the greater the accumulated additional utility.

Given that merchants verify each ticket immediately when 
received, invalid payments are detected 
instantly. On the other hand, duplicated tickets are not detected  
until they are redeemed (if they win the lottery), which may happen  
after several rounds in \sysname. This means that the additional 
utility gain of ticket duplication will be larger than one obtained by 
issuing invalid payments. Therefore, setting the value of 
the penalty deposit based on the additional utility gain of the former covers the latter as well. 
For this reason, we consider only ticket duplication strategy 
in our analysis. \\

\noindent{\bf Game setup.}
We posit a single player game in which a malicious customer 
applies the ticket duplication strategy. This strategy is defined 
as duplicating a sequence number 
among two or more tickets, up to $m$ tickets, where 
$m$ is the number of beneficiary merchants of an escrow.

Based on the design of \sysname's lottery protocol, these 
duplicated tickets will either all win the 
lottery or all lose because the lottery draw 
does not depend on the ticket recipient address (see 
Section~\ref{lottery-protocol}). This means that duplication among 
fewer than $m$ merchants does not reduce the cheating detection 
probability. Therefore, a rational customer who decides to 
duplicate a specific ticket will 
always duplicate it among all $m$ merchants.

Moreover, under the fixed winning rate approach, a rational customer will 
not duplicate more than $(1- p)$\tktrate 
sequence numbers. This is because a number of $p$\tktrate winning tickets (or 
sequence numbers) will 
be selected at each round, and the rest of the tickets, i.e., $(1- p)$\tktrate tickets, will not win. 
Thus, duplicating more than $(1- p)$\tktrate sequence numbers guarantees that 
a duplicated ticket will win. Exceeding this number means that 
the cheating detection probability will be 1.

As mentioned previously, \sysname requires any customer to specify 
the set of merchants who are beneficiaries of its escrow in advance. 
This is needed to be able to bound the additional 
utility gain of malicious customers~\cite{Chiesa17}. 
If this set is not specified, the additional utility  
cannot be bounded because we would not know the maximum number of 
duplicated tickets that could be issued.

Since \sysname adds the parameter \drawlen that groups several 
rounds together when entering the lottery, we refer to a list of contiguous \drawlen rounds 
as a lottery round. Hence, the number of lottery rounds in an escrow 
lifetime is $\frac{l_{\emph{esc}}}{draw_{\emph{len}}}$, 
and in each lottery round a customer can issue up to \drawlen\tktrate tickets\footnote{In case of 
a non-integer number of lottery rounds, we take the ceiling. This makes our bound more 
conservative, and hence, provides stronger motivation for acting honestly.}. 
Based on these variables, the cheating detection period of all duplicated tickets issued 
in any lottery round in \sysname is \ddraw + \dredeem rounds (the latter is in terms of 
simple rounds, where a round is the time needed to mine a block on the blockchain). 
That is, a malicious customer will be 
detected when any of the duplicated tickets is presented to the 
miners\footnote{We do not assume that merchants exchange any 
information about tickets they received.}. This happens if  
duplicated tickets win the lottery and are then 
claimed by the merchants. Considering the worst case 
that this claim may take place in the last round 
of the redeem period, cheating will be detected after 
\ddraw + \dredeem rounds from a ticket issue time. At that time, the 
miners revoke the penalty escrow and 
the cheating customer leaves the system. Otherwise, if none 
of the duplicated tickets win, this customer stays and may 
continue cheating. To simplify the analysis, we will 
use lottery rounds to express the cheating detection period. That is, this period is 
expressed as (\ddraw + \dredeem)/\drawlen lottery rounds. To ease the discussion, 
we use round to refer to a lottery round in the rest of this section.

\begin{table}[t!]
\caption{Notations I.} 
\label{notations}
\centering 
\small{
\begin{tabular}{| p{0.1\columnwidth} | p{0.82\columnwidth} |}\hline\hline
{\bf Symbol} & {\bf Meaning}  \\ [0.5ex] \hline\hline

$\mathcal{C}$ & Honest customer.  \\[0.5ex] \hline
$\mathcal{\hat{C}}$ & Malicious customer.    \\ [0.5ex]  \hline   
$u(\cdot)$ & Utility gain function.    \\ [0.5ex]  \hline  
$\tau$ & The number of tickets that can be issued per lottery round, such that $\tau =$\drawlen\tktrate and $\tau \in \mathbb{N}$. \\ [0.5ex]  \hline
$d$ &  The lottery draw period measured in lottery rounds, such that $d = \frac{d_{\emph{draw}}}{draw_{\emph{len}}}$ and $d \in \mathbb{N}$. \\ [0.5ex]  \hline
$r$ & The ticket redemption period measured in lottery rounds, such that $r = \frac{d_{\emph{redeem}}}{draw_{\emph{len}}}$ and $r \in \mathbb{N}$. \\ [0.5ex]  \hline
$y_i$ & Number of duplicated tickets in a lottery round $i$, such that $0 \leq y_i \leq (1-p)\tau$. \\ [0.5ex]  \hline
$m$ & Number of beneficiary merchants, such that $m \in \mathbb{N}$. \\ [0.5ex]  \hline
$p$ & Lottery winning probability, such that $0 \leq p \leq 1$. \\ [0.5ex]  \hline
$\beta$ & Currency value of a winning ticket, such that $\beta \in \mathbb{R}^+$. \\ [0.5ex]  \hline
$k$ & The escrow lifetime measured in lottery rounds, such that $k = \frac{l_{\emph{esc}}}{draw_{\emph{len}}}$ and $k \in \mathbb{N}$. \\ [0.5ex]  \hline

\end{tabular}}
\end{table}

Table~\ref{notations} summarizes the notations we use in 
this section, including shorter 
abbreviations than those used earlier in this paper to 
simplify presentation. \\

\noindent{\bf Utility gain function definition.} We define the utility gain function of 
any customer as the service value minus the payments made to 
merchants. We compute the expected value of this function for an 
honest customer and for a malicious one that 
uses the ticket duplication strategy. In order to deter cheating, we 
require the latter to be less than or equal to the former. This is achieved 
by setting the penalty deposit to be at least equal to the maximum   
additional expected utility gained by cheating. \\

\noindent{\bf Additional utility gain analysis.}
We now state and prove a lower bound for \bpen based on the 
above game setup.  

\begin{theorem}
For the game and escrow setup described above, 
issuing invalid or duplicated lottery tickets is less profitable   
in expectation than acting in an honest way if: 
\begin{equation}\label{penalty-bound}
B_{\emph{penalty}} > (m-1)p\beta\tau \bigg(\frac{1-p}{1-\frac{1}{_{\tau}C_{(1-p)\tau}}} + (1-p)(d-1)+r \bigg)
\end{equation}

\end{theorem}

\begin{proof}
In \sysname, a customer can create an 
escrow with a $k$ round lifetime. All tickets issued in a round enter the lottery 
after $d$ rounds, and all winning tickets will expire after $r$ rounds from the 
lottery draw time. In other words, for each round $i \in \{1, \dots, k\}$, all tickets issued in 
round $i$ will enter the lottery at round $i + d$ and will expire at round $i+d+r$.

During each round of an escrow lifetime, an honest  
customer can issue up to $\tau$ tickets with unique 
sequence numbers. Each ticket has an expected value of 
$p\beta$ coins, which corresponds 
to the service value a customer obtains from a merchant for handing out this ticket.  
We use this service value in computing the utility gain function, and 
hence, deriving a lower bound for \bpen.

In contrast, when applying the duplicated ticket issuance strategy, for each round 
$i \in \{1, \dots, k\}$ a malicious customer would 
decide to duplicate $y_i$ tickets, where $y_i \in \{1, \dots, (1-p)\tau\}$. If none of the  
duplicated tickets win, which happens with probability 
$\frac{_{(1-p)\tau}C_{y_i}}{_{\tau}C_{y_i}}$,\footnote{The expression $_aC_b$ 
stands for ${a \choose b}$.} the customer 
stays in the system and obtains an additional utility gain of $(m-1)p\beta y_i$ 
over what an honest customer would receive. This occurs because an honest customer will use a  
sequence number with one merchant only. The malicious customer uses a sequence 
number with $m$ merchants, and hence, it obtains 
service from an additional $m-1$ merchants by duplicating a sequence number.

On the other hand, if any of the duplicated tickets wins the lottery at round $i + d$, 
which happens with probability 
$1 - \frac{_{(1-p)\tau}C_{y_i}}{_{\tau}C_{y_i}}$, 
the malicious customer will be detected at round $i+d+r$ (the latest). This 
reduces its additional utility by \bpen since the penalty escrow will be 
revoked by the miners.

Note that when a duplicated ticket wins, meaning that cheating will be detected,  
the malicious customer will still have $r$ rounds to issue tickets. Therefore, 
as a rational behavior, this customer will choose to duplicate 
all tickets in these rounds because it must leave the system either way.

In order to compute the additional utility gain, we need to model the 
duplication decisions a malicious customer would make at each round 
of an escrow lifetime. We use a 
decision process diagram that captures a process evolution over time. 
Such a diagram contains states indicating the rounds, transition probabilities 
between these states computed based on the decisions made at each 
state, and the additional utility of being at each state, or round, in 
the system.

\begin{figure}[t!]
\centerline{
\includegraphics[height= 1.4in, width = 1.0\columnwidth]{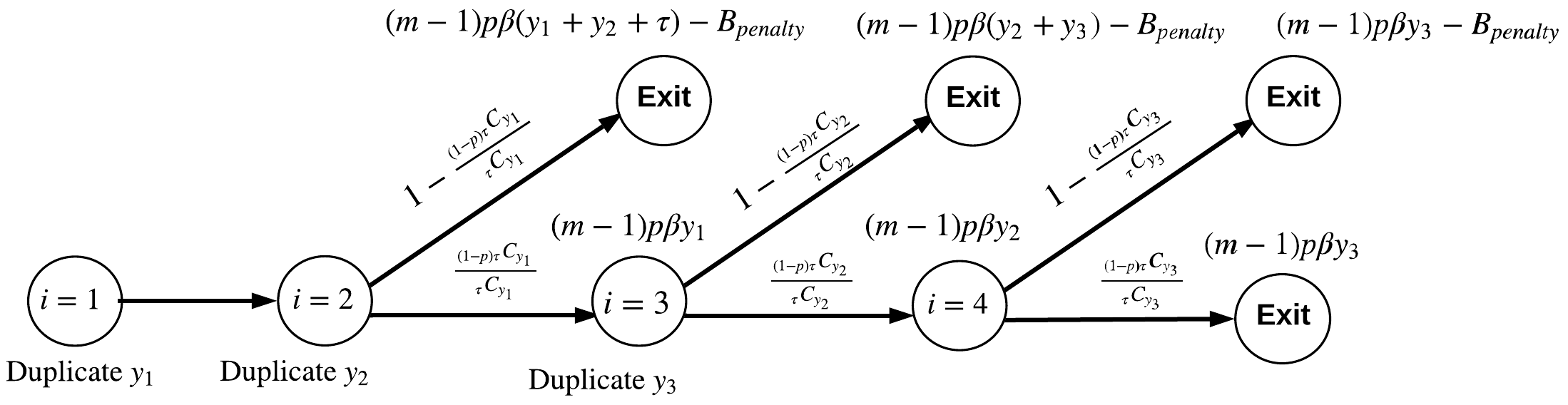}}
\vspace{-7pt}
\caption{Decision process for a $3$ round escrow with $d=2$ rounds and $r=1$ round. Arrows 
carry probabilities, decisions are found below the states, and 
the utility gain is found above the states.} \label{penalty-decision-d2-z1}
\end{figure}

To clarify this concept, we consider a simple case where we have an 
escrow with 3 round lifetime, $d = 2$ rounds, 
and $r = 1$ round. The decision process for this setup is captured in   
Figure~\ref{penalty-decision-d2-z1}. As shown, a customer issues tickets 
for rounds 1 and 2 before any lottery draw takes place, in which  
it duplicates $y_1$ and $y_2$ tickets, respectively. All tickets issued during  
round 1 enter the lottery at the beginning of round 3 (or immediately 
after the end of round 2 as depicted in the figure). If none of these tickets win, the 
malicious customer obtains an additional utility gain of $(m-1)p\beta y_1$ and 
proceeds to round 3. For this round, the customer decides to duplicate 
$y_3$ tickets. On the other hand, 
if any of the $y_1$ tickets wins, the customer knows that it will be detected at 
the end of round 3 (since $r = 1$). Thus, it decides to duplicate all tickets 
in round 3 (i.e., $y_3 = \tau$). The total additional utility the customer obtains in this case, which 
is displayed above the exit state since this customer will leave the system,  
is the sum of the utility gain of duplicating $y_1$, $y_2$, and $y_3$ 
tickets, where $y_3 = \tau$, minus the penalty deposit that will be revoked.

The same analogy is applied to the rest of the rounds, with the exception that at the 
very last rounds there are less than $r$ rounds that can be used at the exit state. 
In other words, the number of remaining rounds 
in the escrow lifetime could be less than $r$, and hence, a customer 
will duplicate fewer than $r\tau$ tickets. This is illustrated in the 
exit state after round $i = 3$ in Figure~\ref{penalty-decision-d2-z1}.

Instead of analyzing a decision process for a $k$ round escrow 
directly, we formulate the 
expected utility gain of a malicious customer in a recursive way. 
That is, we use the expected utility gain in a $k-1$ round escrow 
to compute the expected utility gain in a $k$ round 
escrow, and so on. Intuitively, during the first round of a $k$ 
round escrow, a malicious customer will decide to duplicate $y_1$ 
tickets. If any of these tickets wins at round $1 + d$, cheating will 
be detected. In this case, and as mentioned earlier, the customer 
will duplicate all tickets for the next $r$ rounds and will pay the 
penalty \bpen. This means that with probability $1 - \frac{_{(1-p)\tau}C_{y_i}}{_{\tau}C_{y_1}}$, the 
utility gain is $(m-1)p\beta\big(\sum_{i=1}^d y_1 + r\tau\big) - B_{\emph{penalty}}$.

If none of the $y_1$ tickets 
wins the lottery, the customer stays in the system. In this case,  
round 2 offers a fresh start in a $k-1$ round 
escrow. That is, after collecting the utility gain of duplicating 
$y_1$ tickets, the customer is just like starting fresh in a one round shorter  
escrow. This means that with probability $\frac{_{(1-p)\tau}C_{y_1}}{_{\tau}C_{y_1}}$, the 
utility gain will be $(m-1)p\beta y_1 + \mathbb{E}_{k-1}[u(\mathcal{\hat{C}})]$, 
where the second term denotes the expected utility gain of a malicious customer 
in a $k-1$ round escrow.

Based on the above, we can express $\mathbb{E}_k[u(\mathcal{\hat{C}})]$, 
which is the quantity of interest, as follows:
\begin{multline}
\mathbb{E}_k[u(\mathcal{\hat{C}})] = \bigg(1-\frac{_{(1-p)\tau}C_{y_1}}{_{\tau}C_{y_1}} \bigg)\bigg((m-1)p\beta \sum_{i=1}^d y_i + (m-1)p\beta r\tau - B_{\emph{penalty}}\bigg) + \\ \bigg(\frac{_{(1-p)\tau}C_{y_1}}{_{\tau}C_{y_1}} \bigg)\bigg( (m-1)p\beta y_1 + \mathbb{E}_{k-1}[u(\mathcal{\hat{C}})] \bigg)
\end{multline}

But we have $\mathbb{E}_{k-1}[u(\mathcal{\hat{C}})] < 0$ since the penalty for 
a $k-1$ round escrow has been configured in a way that makes 
$\mathbb{E}_{k-1}[u(\mathcal{\hat{C}})] < 0$ in order to deter cheating. Hence, and 
by requiring $\mathbb{E}_k[u(\mathcal{\hat{C}})] \leq 0$ 
to make cheating unprofitable, we find that: 
\begin{equation}\label{temp}
B_{\emph{penalty}}(y_1, \dots, y_d) > (m-1)p\beta \Bigg(\frac{y_1}{1-\frac{_{(1-p)\tau}C_{y_1}}{_{\tau}C_{y_1}}} + \sum_{i=2}^d y_i + r\tau \Bigg)
\end{equation}

For any $d$ and $r$ value, the above quantity is maximized when  
$y_i = (1-p)\tau$ for $i \in \{1, \dots, d\}$.\footnote{This is done by considering the terms in 
equation~\ref{temp}. For the first term, we used a simple iterative program  
to compute the value of $y_1$ that maximizes this term for $p \in [0.000001, 0.999999]$, 
with an increment step of $0.000001$,    
and $\tau \leq 10^6$. We found that $y_1 = (1-p)\tau$ for all these $p$ values. The second  
term is maximized when $y_i$ is set to its maximum value, which is $(1-p)\tau$, for all 
$i \in \{2, \dots, d\}$.} Substituting these in equation~\ref{temp} produces the lower 
bound stated in the theorem, which completes the proof.\footnote{Equation~\ref{penalty-bound} in 
Theorem 1 is 
reported in Section~\ref{design} as equation~\ref{penalty} after converting the lottery rounds 
into simple rounds using the original notations found in that section.} \qed
\end{proof}

As an example, consider an escrow with a 200 round lifetime, 
$\tau = 1000$ tickets, $p = 0.01$, $\beta = 1$ coin, $m=5$, $d = 6$, 
$r=6$, and \drawlen = 1. Applying equation~\ref{payment-balance} 
produces $B_{\emph{escrow}} = 2,000$ 
coins, and applying equation~\ref{penalty-bound} 
produces $B_{\emph{penalty}} > 477.6$ coins.

\section{Security Analysis}
\label{security-analysis}
\vspace{-8pt}
In this section, we analyze the resilience of \sysname to the threats outlined in 
Section~\ref{threat-model}. To defend against these threats, our scheme  
utilizes cryptographic and financial techniques based on the threat type to 
be addressed. This is discussed in the following paragraphs. \\

\noindent{\bf Escrow overdraft.} 
This threat can be exploited using several strategies, including: 
\begin{itemize}
\itemsep0em
\item A customer creates a payment escrow with a balance that 
cannot cover all winning tickets, or a penalty escrow with a balance that 
cannot cover the financial punishment.
\item A customer issues more tickets than the \tktrate specified in 
its escrow setup.
\item A customer performs a front running attack in which it 
withdraws the payment escrow before paying merchants, or 
withdraws the penalty escrow before paying the 
financial punishment when caught cheating. 
\end{itemize}

The first strategy, in which a customer creates 
payment and penalty escrows with insufficient balances, is neutralized 
by the escrow setup of \sysname. When processing an escrow 
creation transaction, the miners check that the payment escrow balance 
covers all winning tickets (see  
Section~\ref{escrow-setup}). In addition, they check that the penalty 
deposit meets the lower bound derived in Section~\ref{penalty-analysis}. 
The miners will reject any escrow that does not satisfy these conditions.

The second strategy, i.e., issuing more tickets than can be 
covered by the escrow,
cannot be performed because lottery tickets are tracked using their 
sequence numbers. When receiving a ticket, a merchant checks that a sequence number to be 
within the range assigned to a ticket 
issue round before accepting the ticket. Hence, if a customer exceeds \tktrate in 
any round, merchants will detect that immediately. Merchants will reject any 
ticket outside of the current round (modulo one round to deal with 
desynchronized views of the blockchain).

As for the last strategy that covers front running attacks, such 
attacks are mitigated by the 
escrow spending mechanism and the lottery protocol implemented in 
\sysname. A customer does not control any of the escrows it owns. Instead, 
fund release is triggered only by the receipt of a valid winning lottery ticket, 
in the case of a payment escrow, or a valid proof-of-cheating, in the case of 
a penalty escrow. Honest miners will enforce these rules in the system.\\

\noindent{\bf Duplicate Ticket Issuance.} 
\sysname addresses this attack financially through a detect-and-punish 
approach. Any party that detects two or more tickets issued using the same escrow
and carrying identical 
sequence numbers can produce a proof-of-cheating 
against the issuing customer. Miners publish this proof on the blockchain, which 
burns the customer’s penalty escrow as a punishment. Setting the penalty deposit 
as described in Appendix~\ref{penalty-analysis} makes cheating unprofitable, 
which deters rational customers from attempting this malicious strategy. \\

\noindent{\bf Invalid payments.}
To pursue this attack, a customer may issue tickets with invalid 
format or invalid field values knowing that these tickets cannot be 
claimed later if they win. An invalid ticket is one that uses  
an invalid escrow (e.g., a broken one), has  
an invalid $(t_{\emph{issue}}, seqno)$ tuple, or contains a 
merchant index that is not listed in the escrow state. These events can be detected 
instantly because merchants validate all lottery tickets before 
accepting them. As mentioned previously, an invalid ticket is  
dropped unless it has an invalid $(t_{\emph{issue}}, seqno)$ tuple. 
As such a ticket can be used as a proof-of-cheating, the customer  
loses the penalty escrow, which exceeds any  
profit from cheating. This discourages rational customers from issuing   
invalid tickets. \\

\noindent{\bf Unused-escrow withholding.}
This threat is mitigated by the expiration rule of lottery tickets and 
\sysname's escrow refund policy. Each ticket has a set redemption 
period after which it expires. Hence, a   
merchant who tries to put an escrow on hold by delaying a winning ticket 
claim will lose its payment. Furthermore, when all 
tickets tied to an escrow expire, i.e., when \trefund is approached, the 
miners will allow the customer to spend the remaining funds 
in its escrow. This prevents locking 
unused escrow funds indefinitely on the blockchain. \\

\noindent{\bf DoS.} As DoS is a large threat category, in this 
work we consider only the cases that are unique to the design of  
\sysname. These include preventing customers from creating 
escrows, preventing 
merchants from claiming their winning lottery tickets, or controlling   
relaying blocks and transactions based on their content. Such  
situations may happen when miners disregard specific 
transactions or blocks, or 
when an attacker controls the network links and drops specific 
transactions or blocks.

The case of miners disregarding specific transactions/blocks may take place 
when an attacker controls a substantial portion of the mining 
power. This may work even under the assumption that the 
majority of the mining power is honest. That is, an attacker 
with $< 50\%$ of the mining power may still 
be able to perform harmful attacks, e.g., selfish mining~\cite{Sapirshtein16}.  
To protect against selective relaying, techniques that allow propagating blocks 
and transactions without disclosing  their content  
can be employed, e.g., BloXroute~\cite{bloxroute}.

The case of controlling the network links, which represents 
attacks against network availability, 
is a potential problem in any distributed system. 
Deploying mechanisms to enhance network  
connectivity, such as participants maintaining connections with large 
number of miners, may reduce the impact of this attack. Such 
mechanisms are independent of the design of \sysname, and so it is 
up to the parties themselves to adopt suitable solutions. \\

\noindent{\bf Lottery manipulation.} 
This threat covers all strategies that could be used to manipulate the 
lottery draw, including:
\begin{itemize}
\itemsep0em
\item An attacker, who could be any party inside or 
outside the system, tries to influence the 
hash used in a lottery draw in order to 
make specific tickets win or lose. 

\item A malicious customer may issue winning lottery tickets, to itself or 
to malicious colluding merchants, to drain the escrow 
before other merchants can claim their winning tickets.

\item A malicious customer deliberately issues  
losing tickets to merchants to avoid paying them.

\item An attacker, insider or outsider, tries to issue lottery 
tickets to herself or others.
\end{itemize}

In the first strategy, an attacker tries to influence the lottery by 
controlling the hash used in the lottery draw. This can be done by 
either manipulating the ticket fields that impact the lottery, or by 
controlling the hash of the block mined at time \tdraw. In the former, 
the issuing customer may tweak a ticket in order to influence its 
chance of winning the lottery. In the latter, a 
miner may forgo any block that does not produce a favorable 
lottery outcome (i.e., a favorable VDF value), or even publish an 
incorrect VDF value in order to bias the lottery draw outcome.

All these malicious behaviors are mitigated by the lottery protocol design. 
The ticket fields that impact the lottery draw include only the ones that 
appear in the escrow state, and these cannot be tweaked by the issuing 
customer, or by any other party. As for discarding unfavorable blocks, recall that 
the lottery draw is based on the VDF value of the block at index \tdraw.  
Therefore, a miner who  
chooses to perform this computation and then announce a favorable block, 
is not likely to succeed in publishing this block on the 
blockchain. Computing the VDF takes time, and in the meantime 
other miners will announce their blocks immediately. 
These blocks are more likely to be adopted by the network, and hence, to be 
used in the lottery draw (under the 
assumption that the majority of the mining power is honest). Furthermore, 
publishing an invalid VDF value will not succeed because honest miners, 
who already computed this value while mining their blocks, will 
reject a newly received block with a VDF value that does not 
match their own. Consequently, 
such a malicious miner will not only fail in biasing the lottery, but also will 
lose the mining rewards.

In the second strategy, a malicious customer tries to withdraw an escrow 
indirectly by issuing winning tickets to itself, and claiming  
these tickets before merchants can claim their payments. To do so, a 
customer may wait until the block at index \tdraw is mined and then 
print winning tickets to itself. This technique is mitigated by the lottery protocol design. 
As the ticket issue schedule specifies both issue and lottery draw time for each round, 
it will be too late for the customer to select winning sequence numbers after \tdraw. 
After observing the block at time \tdraw, this customer cannot do anything 
to produce winning tickets other than checking which sequence numbers 
are winning (assuming it did not use these sequence numbers to pay 
any merchant earlier). As mentioned previously, this customer cannot manipulate the ticket fields to 
make a losing sequence number win. Also, it cannot issue losing 
tickets using these sequence numbers to merchants because the 
issue time of these tickets was at least \ddraw rounds earlier, and hence, merchants 
will not accept these tickets that lost the lottery. On the other hand, should a customer who saves 
some sequence numbers and uses them to issue tickets to itself have winning tickets, it will 
not affect the payments of other merchants. Such sequence 
numbers represent legitimate tickets and are therefore are covered by the 
payment escrow balance.

The third strategy, where a customer deliberately hands merchants 
losing tickets, will fail with overwhelming probability. In order to determine 
which ticket will lose the lottery, the customer 
needs to either predict the hash of the future block mined at time \tdraw in order to compute 
its VDF value, or take a guess at the VDF value itself. Since hash functions are 
modeled as random oracles, predicting such values succeeds with negligible probability.

Lastly, the fourth strategy, in which an attacker tries to issue tickets tied 
to escrows she does not own, will not succeed due to the system's use of 
secure cryptographic signatures. \sysname requires customers to sign 
all lottery tickets they issue, which means that to issue a valid ticket an attacker 
needs to forge the customer’s signature. By using a secure digital signature 
scheme such an attack will fail with overwhelming probability.

\section{Performance Evaluation}
\label{perf-eval}
In order to understand the performance benefit of concurrent probabilistic 
micropayments, this section evaluates the computation, bandwidth, and 
payment setup costs of \sysname. To do this, we conduct 
empirical experiments to answer the following questions:
\begin{itemize}
\itemsep0em
\item How fast can customers, merchants, and miners 
process lottery tickets? 

\item What is the bandwidth cost of exchanging these 
tickets?

\item What is the size of escrows on the blockchain?

\item How do the schemes compare using workload numbers 
derived from real world scenarios?
\end{itemize}

To put our results in context, we compare our scheme with 
MICROPAY~\cite{Pass15}. The rest of this section describes our 
methodology and discusses the significance of the obtained results.

\subsection{Methodology}
To establish our benchmarks, we implemented the functions 
used for generating tickets, verifying these ticket, and 
performing a lottery draw. For \sysname, we followed the design 
introduced in this paper. For MICROPAY, we tested    
its fully decentralized version, called MICROPAY1, 
with its non-interactive lottery protocol as outlined in~\cite{Pass15}. This protocol 
requires a merchant to publish the description of a verifiable unpredictable 
function to perform the lottery. For this function, we used the 
verifiable random function (VRF) construction introduced by Goldberg 
et al.~\cite{Goldberg16} with its implementation over the NIST P-256 
curve~\cite{ecvrf}.

Two cryptographic primitives affect the implementation of both \sysname 
and MICROPAY, 
namely, hash functions and digital signatures. For hashing, we used SHA256. 
For digital signatures, we chose to test the most common elliptic curve based 
constructions. These include ECDSA with secp256k1 curve (used in Bitcoin and most 
cryptocurrencies), ECDSA with P-256 curve (widely used and  
recommended by NIST), 
and EdDSA with Ed25519 curve~\cite{Bernstein12} (of a great 
interest recently as its security and efficiency have encouraged several major cryptocurrencies 
to either use this scheme~\cite{stellar-eddsa,monero-eddsa} or prepare to switch to  
it~\cite{ripple-eddsa,ethereum-eddsa}).

We computed the performance metrics of interest as follows. 
The computation 
cost was measured as the rate at which customers, merchants, and miners can 
process lottery tickets. Bandwidth overhead was calculated by 
reporting on the size of tickets (in bytes) when exchanged between customers 
and merchants, and when claimed through the miners. To evaluate the 
effect of micropayment concurrency, we computed the number of 
escrows a customer would need  
to support the ticket issue rate in each of the tested schemes. 
Lastly, we studied two real life applications 
and computed the overhead of processing micropayments using 
workload numbers derived from these applications.

Our experiments were implemented in C on an Intel Core i7-4600U 
CPU @ 2.1 GHz, with 4 MB cache and 8 GB RAM, where each of the payment 
processing functions was called $10^6$ times.

\subsection{Microbenchmark Results}
\vspace{-2pt}
\subsubsection{Lottery ticket processing rate} 
\label{processing-rate}
We start by quantifying the 
computation cost of processing micropayments in both schemes. 
This is done by measuring the rate 
at which a customer can generate lottery tickets, the rate 
at which a merchant can process these tickets, which involves 
both validating a ticket and running the lottery\footnote{Although 
a merchant in \sysname runs the lottery several rounds after  
validating a ticket, we report the cost of these two operations together. 
This is because such cost is the overhead per ticket incurred on the 
merchant side.}, and the rate at which miners validate claimed tickets and 
running the lottery for the escrows. 
The obtained results are found in Table~\ref{tkt-rate}.

\begin{table}[t!]
\caption{Ticket processing rate (ticket / sec).} 
\label{tkt-rate}
%\vspace{-10 pt}
\centering \scriptsize{
\begin{tabular}{| p{0.12\columnwidth}  | p{0.14\columnwidth} | p{0.12\columnwidth} | p{0.14\columnwidth}| p{0.14\columnwidth} | p{0.12\columnwidth} | p{0.14\columnwidth}|}\hline\hline
&\multicolumn{3}{|c|}{\bf MICROPAY} &\multicolumn{3}{|c|}{\bf \sysname} \\ [0.5ex] \hline
& {\bf ECDSA (secp256k1)} & {\bf ECDSA (P-256)} & {\bf EdDSA (Ed25519)} & {\bf ECDSA (secp256k1)} & {\bf ECDSA (P-256)} & {\bf EdDSA (Ed25519)} \\ [0.5ex] \hline\hline
 
Customer & 1,859 & 32,471 & 26,238 & 1,868 & 33,006 & 26,749 \\[0.5ex] \hline

Merchant & 1,328 &  2,399 & 2,561 & 2,249 & 10,505 & 8,473 \\ [0.5ex]  \hline   

Miner & 1,340 & 2,448 & 2,617 & 2,241 & 10,345 & 8,368 \\ [0.5ex]  \hline    

\end{tabular}}
\vspace{-10 pt}
\end{table}

As the table shows, customers in both schemes 
generate tickets at comparable rates   
because the operations performed  
are almost identical in \sysname and MICROPAY. Given that the 
heaviest operation in this process is signing a ticket, the generation rates 
improve by using an efficient digital signature scheme (when 
replacing ECDSA (secp256k1) 
with ECDSA (P-256) and EdDSA (Ed25519), performance 
is boosted by around 17x and 14x, respectively).

The trend is different for the rates of merchants and miners. These parties in \sysname 
are 1.7x, 4.2x, and 3.2x faster than in MICROPAY for the three digital 
signature schemes. This is because miners and merchants  
run and verify the lottery draw outcome. In \sysname,  
this process involves only lightweight hash operations. 
On the other hand, the lottery in MICROPAY requires evaluating, and 
proving the output  correctness, of a computationally-heavy VRF.

Furthermore, merchants and miners in \sysname benefit more from 
the efficiency of the used digital signature scheme. This is because the heaviest 
operation these parties perform when processing a ticket in \sysname is verifying 
a customer's signature. However, in MICROPAY the bottleneck 
is evaluating a VRF and producing a proof that its output is correct on the merchant side, and verifying 
this proof on the miner side. As shown in the table, 
MICROPAY obtains only around 1.9x improvement when replacing ECDSA (secp256k1) 
with any of the other two schemes. In contrast, \sysname achieves around 4.7x and 3.8x 
improvement when replacing ECDSA (secp256k1) with ECDSA (P-256) or EdDSA (Ed25519),
respectively.  

{\bf Key Takeaway:} Compared to MICROPAY, \sysname reduces the computational 
load on merchants and miners by a factor of 1.7-4.2x.

\subsubsection{Lottery ticket bandwidth cost} 
\label{bw-cost} 
In terms of bandwidth, 
\sysname incurs less overhead than MICROPAY because its lottery tickets 
are smaller. A ticket sent from a customer to a merchant is 110 bytes in  
%Ghada: signature 72 bytes, size of signature 1 byte, escrow id 32 byte, merchant index 1 byte
\sysname, while it is 274 bytes in MICROPAY. 
A winning ticket sent from merchants 
to miners is also 110 bytes in \sysname, while it is 355 bytes for MICROPAY  
because this ticket must be accompanied with a  
NIZK proof. This means that \sysname incurs only 40$\%$ of 
the bandwidth overhead 
of MICROPAY between customers and merchants, and only 31$\%$ of the 
overhead between merchants and miners.

To put these numbers in context, we report on the transaction sizes in Bitcoin. 
The average size is around 500 bytes, where a transaction with one or two inputs 
is about 250 bytes~\cite{bitcoin-tx-size}. Adding a winning ticket as one of 
these inputs produces a 
claim transaction with a size of 360 bytes in \sysname, which is less than the 
average Bitcoin transaction size. On the other hand, in MICROPAY the size of a claim 
transaction will be 605 bytes, exceeding the average size. 

{\bf Key Takeaway:} The use of efficient lottery protocol reduces not only 
the computation cost in \sysname, but also the bandwidth cost of exchanging 
lottery tickets and the amount of data to be logged on the blockchain.

\subsubsection{Size of escrows on the blockchain}
\label{escrow-size}
One major difference between concurrent and sequential micropayment schemes
is that the former need a new escrow after each winning ticket.
Additionally, to issue tickets in parallel at a fast rate in sequential schemes, 
the customer needs a large number of escrows. This is because a new ticket cannot 
be issued using the same escrow until it is confirmed that the prior ticket 
did not win, which requires the customer 
to wait for the merchant to announce the lottery result. Hence, even if the customer is 
capable of generating tickets at a fast rate, it might slow down just because 
it does not have enough escrows to allow this rate. Furthermore, even if the 
customer is willing to create larger number of escrows, this dramatically increases the 
overhead. Each of these escrows requires an individual 
escrow creation transaction, which in turn requires paying a 
transaction fee and logging on the blockchain.

For example, to support the ticket issue rates reported in 
Table~\ref{tkt-rate}, a MICROPAY customer would need a large number of escrows. 
The exact number of escrows needed depends on the network 
latency and a merchant's ticket processing rate. Using the average US RTT of 
31 ms~\cite{att}, and the processing time of the tickets, 
in the best case an escrow in MICROPAY can be used to issue  
32 tickets per second (this is in case none of these ticket win or only the 
last one wins). Therefore, a customer in MICROPAY would need 
60, 1019, or 653 escrows per second to support the generation rates for signature schemes  
ECDSA (secp256k1), ECDSA (P-256), or EdDSA (Ed25519), respectively, 
as found in Table~\ref{tkt-rate}. On the other hand, a customer in \sysname 
would need only \emph{one} escrow with the proper balance to pay at any 
given ticket rate. As such, \sysname dramatically reduces the 
amount of data logged on the blockchain.

{\bf Key Takeaway:} Supporting micropayment concurrency  
dramatically reduces the amount of escrow data on the blockchain.

\subsection{Micropayments in Real World Applications}
\label{micro-real-world}
\vspace{-4pt}
To ground our results in real world numbers, we examined two 
micropayment applications: online gaming and peer-assisted 
content delivery networks (CDNs). We computed the overhead of 
processing micropayments with parameter values derived based on the service 
price and workload in these applications. This is done for three cases: 
Bitcoin without employing any micropayment scheme, Bitcoin 
with MICROPAY, and Bitcoin with \sysname.
%Ghada: all numbers are up to 6 figures

\vspace{-8pt}
\subsubsection{Setup}
To compute the overhead, we estimated the service costs and loads from real
world sources.  For online gaming, we used data from the 
popular game Minecraft~\cite{minecraft} as 
an example. The average mid tier cost of playing this game for 8 players 
is around \$12 per 
month~\cite{minecraft-price}. We considered 1000 players distributed among 
125 servers. This means that the service price is \$0.034722 
per minute (or \$0.000579 per second) for the 1000 players.

For the peer-assisted CDN application, we considered a content publisher 
that hires peers as caches to distribute the content for its clients.  
Suppose a publisher wants to serve content at roughly 1Gpbs. 
Such a service costs around \$17,312 monthly in the 
US~\cite{cdn-price}, and hence, on average, the service cost per second 
is \$0.006679. The publisher will provide a lottery ticket to a cache
for each 1MB data chunk it serves. Thus, to support a rate of 1Gbps,  the 
publisher will issue 128 tickets per second.

With these combined values, it is possible to compute the lottery 
winning probability $p$ and the currency value of a winning ticket 
$\beta$.  The former is done by determining the total transaction fee 
to be paid for the miners (per second), and then computing $p$ in a way that ensures 
the fees paid when claiming winning tickets (per second)  
do not exceed this value. We 
consider the fee to be 2\% of the service cost paid per 
second~\cite{additonal-fee} (at a minimum)\footnote{In both examples, we 
assume that the players and the publisher will pay at the 
same price offered by a gaming or CDN company.}. For the fee of 
redeeming a winning ticket, we use the median transaction 
fee in Bitcoin, which is 
around \$0.068 as of late January 2019~\cite{bitcoin-tx-fee}.

Based on the above, the fees per second are equal to 
the expected number of winning tickets per second multiplied by the 
transaction fee that the miners charge. So in Minecraft, $p$ can be computed as 
$p = \frac{(0.02)(0.000579)}{(16.67)(0.068)} = 0.00001$, where 16.67 is 
the number of tickets issued by all players per second (the 1000 players 
issue 1000 tickets per minute). Similarly, a publisher in 
our CDN example can compute $p$ as 
$p = \frac{(0.02)(0.006679)}{(128)(0.068)} = 0.000015$.

As for computing $\beta$, it can be estimated by dividing the service 
cost by the number of winning tickets (both per second). In Minecraft, 
this produces $\beta = \$3.472$, and it is \$3.4 in the CDN application.

For the escrow setup, in Minecraft, we assume that each player creates 
one escrow per month. For the CDN application, we consider a publisher 
creating one escrow per day. In addition to that, in 
MICROPAY a new escrow must be created after each winning ticket\footnote{To 
simplify the discussion in this section, and since the goal is to evaluate the  
financial and bandwidth cost of the payment setup rather than its computational 
cost, we allow \sysname 
to operate with a non-integer number of winning tickets per round. This is done 
by adopting the lottery protocol with independent ticket winning events described in 
Appendix~\ref{lottery-protocol-ind}.}.

It should be noted that in both the gaming and CDN examples, we only account for the cost 
of operating the service. Factors such as funding  
Minecraft's development team (which was presumably supported by the initial 
purchase of the game) or produce the video content that was served by the 
CDN are not considered here. In practice, operational costs are minimal 
compared to the development cost which 
can run in the hundreds of millions of dollars~\cite{game-dev-cost}.

We used this setup to compute overhead of micropayment 
processing for both applications, as shown in what follows.

\vspace{-8pt}
\subsubsection{Online Gaming Application}
We used the configuration parameters outlined earlier 
to compute the transaction fees and the bandwidth cost of micropayment 
processing in this application. The results are found in Table~\ref{results-game}.

\begin{table}[t!]
\caption{Micropayment overhead in online gaming (a round is 10 min).} 
\label{results-game}
\centering \small{
\begin{tabular}{| p{0.4\columnwidth}  !{\vrule width 2pt} p{0.19\columnwidth} | p{0.19\columnwidth} | p{0.19\columnwidth} |}\hline
Metric  & {\bf Bitcoin} & {\bf MICROPAY} & {\bf \sysname}  \\ [0.5ex] \hline\hline

Winning tickets / sec  & N/A  & 0.000167 & 0.000167    \\[0.5ex] \hline

Escrows / sec & N/A  & 0.000552 & 0.000386   \\ [0.5ex]  \hline   

Transactions /sec  & 16.67  & 0.000719 & 0.000552   \\ [0.5ex]  \hline   

Transaction fees / round & \$680  & \$0.029341 & \$0.022541   \\ [0.5ex]  \hline  

Bandwidth between customers and miners & 3,333 bps  & 1.105 bps & 1.009 bps  \\ [0.5ex]  \hline   
 
Bandwidth between customers and merchants & N/A  & 36,533 bps & 14,667 bps  \\ [0.5ex]  \hline   

Bandwidth between merchants and miners & N/A  & 0.807 bps & 0.523 bps  \\ [0.5ex]  \hline

Delta blockchain size / round & 2.38 MB  & 0.000137 MB & 0.00011 MB  \\ [0.5ex]  \hline

\end{tabular}}
\vspace{-7 pt}
\end{table}

We start with the number of transactions 
the miners process. In Bitcoin, all micropayments are processed as
individual transactions. In contrast, with MICROPAY and \sysname, 
only winning tickets and escrow creation transactions will be sent to the 
miners. MICROPAY is a sequential scheme, so every time a 
ticket wins the escrow
breaks. Therefore, the number of escrows per round equals to the expected number of 
winning tickets per round. 
In \sysname, however, a player may use one escrow for the duration of their 
subscription (a month in our case). Consequently, 
and as the table shows, \sysname generates 
25\% fewer transactions, accounting both escrow and winning ticket redemption.

The number of transactions affects the amount of fees miners charge for processing. 
As the table shows, processing micropayments 
individually is expensive, costing \$680 per round (a transaction 
costs around \$0.068 as mentioned previously). However, in \sysname and 
MICROPAY, the fees are much lower because they are paid only when 
claiming winning tickets or creating 
escrows. Furthermore, due to the reduced
number of escrows, \sysname incurs the 
least fees, costing around 75\% of what MICROPAY incurs.

In calculating the bandwidth overhead, in Bitcoin players send all micropayments directly to the 
miners. They do not send anything to the servers nor do the servers send anything to the 
miners. On the other hand, in MICROPAY and \sysname, all lottery tickets are exchanged 
locally between 
players and servers. Only escrows and winning tickets will be sent to the miners. 
Based on the average size of a Bitcoin transaction (around 250 bytes as mentioned earlier), 
the size of an escrow transaction in MICROPAY is around 250 bytes, wheares  
in \sysname it is around 327 bytes. This is because our scheme adds additional fields 
to store the payment setup parameters as described in 
Section~\ref{escrow-setup}. For the size of a claim transaction, beside the 
transaction average size, we 
add the size of a winning ticket (110 bytes in MicroCash and 355 
bytes in MICROPAY). As shown in Table~\ref{results-game}, 
the bandwidth cost between players/servers and the miners in Bitcoin is more than 3000x 
the cost incurred in MICROPAY or \sysname. This shows the great benefit of processing 
payments locally using a micropayment scheme.

The bandwidth cost of the miners can be used to quantify the 
increase in the blockchain size per round since all transactions sent to the miners are 
logged on the blockchain. As the table shows, logging all micropayments is prohibitive as it 
requires more than 2 MB per round. In Bitcoin, only one block of size 1MB 
can be published per round, meaning that paying at this relatively slow rate cannot be supported. 
On the other hand, this overhead is reduced to less than 0.00014 MB in MICROPAY 
and \sysname.

\vspace{-8pt}
\subsubsection{Peer-assisted CDN Application}
Because of the larger workload involved, our results show micropayment 
schemes offer even more dramatic benefits when used to serve CDN traffic.
As Table~\ref{results-cdn} shows, in plain Bitcoin, miners process 128 
transactions per second, which is the number of data chunks caches serve 
per second. This number drops to fractions in MICROPAY, and goes further 
down an 50\% in \sysname. 
This is reflected in both the transaction fees and the blockchain 
size.

\begin{table}[t!]
\caption{Micropayment overhead in peer-assisted CDNs (a round is 10 min).} 
\label{results-cdn}
\centering \small{
\begin{tabular}{| p{0.4\columnwidth}  !{\vrule width 2pt} p{0.19\columnwidth} | p{0.19\columnwidth} | p{0.19\columnwidth} |}\hline
Metric  & {\bf Bitcoin} & {\bf MICROPAY} & {\bf \sysname} \\ [0.5ex] \hline\hline

Winning tickets / sec  & N/A  & 0.001964 & 0.001964   \\[0.5ex] \hline

Escrows / sec & N/A & 0.001976 & 0.000012  \\ [0.5ex]  \hline   

Transactions /sec  & 128 & 0.00394 & 0.001976  \\ [0.5ex]  \hline   

Transaction fees / round & \$5,222 & \$0.160769 & \$0.08062  \\ [0.5ex]  \hline  

Bandwidth between customers and miners & 256,000 bps & 3.95 bps & 0.165 bps  \\ [0.5ex]  \hline   
 
Bandwidth between customers and merchants & N/A & 280,576 bps & 112,640 bps  \\ [0.5ex]  \hline   

Bandwidth between merchants and miners & N/A & 9.508 bps & 6.16 bps  \\ [0.5ex]  \hline

Delta blockchain size / round & 18.31 MB & 0.000963 MB  & 0.000452 MB  \\ [0.5ex]  \hline

\end{tabular}}
\vspace{-7 pt}
\end{table}

Processing micropayments 
individually costs more than \$5,000 per round, while these fees drop to 
cents when a micropayment scheme is employed. It also 
requires logging more than 18 MB per round on the blockchain. 
This overhead is reduced to around 0.001 MB in MICROPAY and 
0.0005 in \sysname. This shows the great advantage of employing a micropayment 
scheme for heavy loaded applications, and the benefit of 
supporting micropayment concurrency (where \sysname reduced the 
additional blockchain size and the total fees by around 50\%).

In terms of bandwidth overhead between participants, in plain Bitcoin the 
miners have at least 19,000x the cost when a micropayment scheme is 
employed. Moreover, \sysname incurs almost no bandwidth cost between the 
publisher and the miners. This is despite the fact that in this application, an 
escrow creation transaction in \sysname is larger than the one needed in 
MICROPAY (such a transaction is around 1,783 bytes in \sysname, where we 
consider 45 beneficiary caches to support the rate of 1Gbps\footnote{We use the
 average upload speed in the US~\cite{upload-speed}, which is 
22.79 Mbps. Hence, to serve 1Gbps, a publisher needs 45 caches.}). 
Such a minimal cost for \sysname is due to payment concurrency as it 
allows creating a one long-lifetime escrow instead of 
large number of escrows, as MICROPAY requires.  Even for ticket 
exchange, \sysname incurs a lower cost although both schemes have the 
same number of tickets. 
This is because a ticket (on-chain or off-chain) in \sysname is substantially
smaller than in MICROPAY.

{\bf Key Takeaways:} Micropayments are absolutely critical to be able to process 
small transactions
in modern applications.  \sysname is cost efficient enough to be used in 
online gaming and content distribution.  Concurrent use of a single escrow 
decreases the total data added to the blockchain by roughly half.

\section{Conclusions}
\label{conclusions}
In this paper, we introduce \sysname, the first 
decentralized probabilistic framework that supports concurrent 
micropayments. The design of \sysname features an  
escrow setup and ticket tracking mechanism that permit a customer 
to rapidly issue tickets in parallel using only \emph{one} escrow. 
Moreover, \sysname is cost effective as it implements a non-interactive 
lottery protocol for micropayment aggregation that is based solely on 
secure hashing. When compared to the sequential scheme MICROPAY, 
\sysname has substantially higher payment processing rates and much lower 
bandwidth and on-chain traffic.
This demonstrates the viability of employing our scheme in large-scale 
micropayment applications.

{\footnotesize \bibliographystyle{acm}
\bibliography{microBib}}

\normalsize
\begin{appendices}
\renewcommand{\thesection}{\Alph{section}}

\section{Lottery Protocol With Independent Ticket Winning Events}
\label{lottery-protocol-ind}
In this section, we describe a variant of our lottery protocol that follows the 
same paradigm found in previous 
works~\cite{Wheeler96,Rivest97,Pass15,Chiesa17}. That is, each lottery ticket 
enters the lottery independently of other tickets. Thus, there is a 
chance that all tickets in a round may win or all lose with an expected number 
of winning tickets of $p$\tktrate per round. This protocol can be used in applications 
where it is infeasible to set \drawlen, $p$, and \tktrate in a way that produces 
an integer number of winning tickets per \drawlen rounds, such as the real life 
applications discussed in Section~\ref{perf-eval}. In 
this section, we first introduce this protocol variant, after which we derive 
bounds for the payment and penalty escrow balances, denoted as 
\bescind and \bpenind, respectively, under the modified protocol setup.

\subsection{Lottery Protocol Design}
The lottery protocol with independent ticket winning events is similar to the 
one introduced in Section~\ref{lottery-protocol} in the sense that it is based 
solely on secure hashing, requires the miners to compute a VDF value that 
is used in the lottery draw, and defines parameters \ddraw and 
\dredeem to control when a ticket can enter the lottery, and when 
it can be redeemed if it wins.

However, this protocol differs from one with an exact win rate 
in three aspects. First, given that each ticket 
enters the lottery independently of others, there is no restriction that $p$\tktrate 
be an integer value. Hence, there is no need for the parameter 
\drawlen. This also affects how the lower bound of 
\bpenind is computed, as we will see in Section~\ref{penalty-analysis-ind}. Second, 
given that there is a chance that more 
tickets than expected (i.e., $>p$\tktrate\lesc) may win, we require that 
the payment escrow cover all winning tickets with very high probability, at least 
$1 - \epsilon$ for some small $\epsilon$ value such as $\epsilon < 0.05$. 
This affects how the value of \bescind is computed as we will see in 
Section~\ref{b-escrow-ind}.  And third, we need a different lottery draw mechanism 
to allow tickets to enter independently. 
We elaborate on this mechanism in what follows.

As before, a merchant keeps a ticket $tkt_L$ until its  
lottery draw time \tdraw, which is exactly \ddraw rounds after its issue time. Then, after 
observing the VDF value of the block 
mined at that time, the merchant computes the following quantity over the  
ticket (where $H$ is a hash function):
\begin{equation}\label{tkt-hash}
h_{tkt_L} = H(id_{\emph{esc}}||seqno||VDF(Block_{t_{\emph{draw}}}))
\end{equation}

A ticket wins if the least significant word of $h_{tkt_L}$ is less than 
$2^{32}p$. Therefore, within the precision allowed by a 32-bit number, a ticket 
will have a $p$ chance of winning.

This process is clarified by the example depicted in 
Figure~\ref{lottery-example-old}. As shown, the ticket has been issued at round 30, and hence, 
it entered the lottery at round 40. The VDF value of the block with index 
40 appears inside block 41. Based on the value of $h_{tkt_L}$, the ticket in the 
figure is a winning one.

\begin{figure}[t!]
\centerline{
\includegraphics[height= 1.6in, width = 0.9\columnwidth]{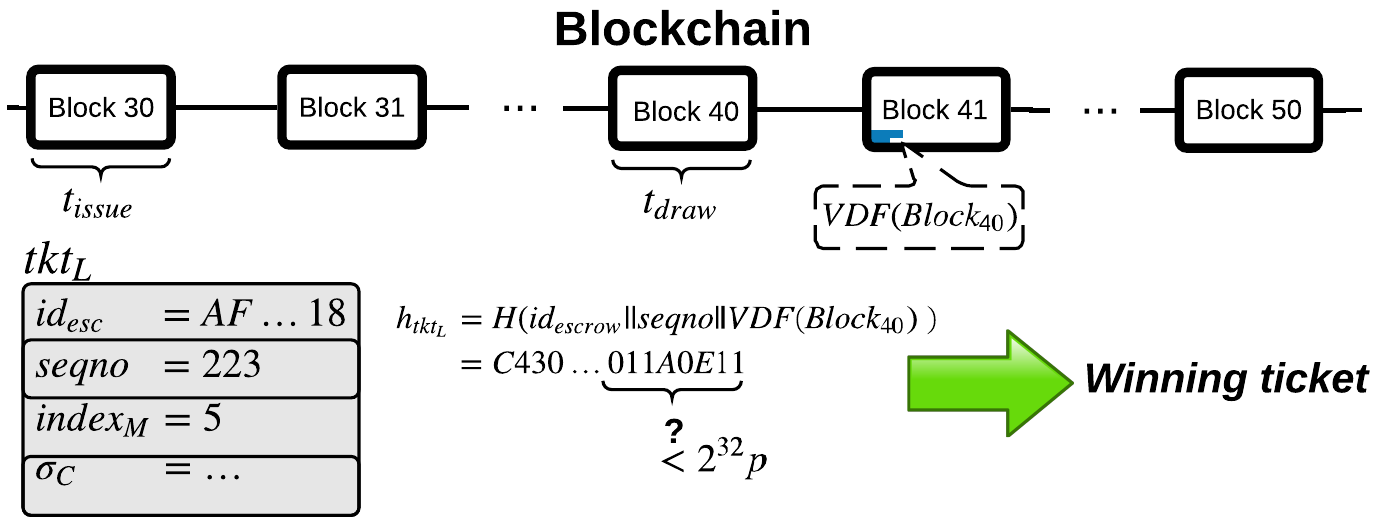}}
\vspace{-5pt}
\caption{Lottery draw example (\ddraw = 10, and $p = 0.01$). } 
\label{lottery-example-old}
\end{figure}

As before, $h_{tkt_L}$ involves only the ticket 
fields that are part of the escrow state that the issuing customer cannot manipulate. 
These fields do not include the merchant recipient address, which 
means that a ticket's chance of winning the lottery is not affected by 
who owns it. Hence, this protocol variant has the same security properties as  
the protocol with an exact ticket win rate.

\subsection{Computing \bescind}
\label{b-escrow-ind}
In this section, we show how to compute the payment escrow 
balance \bescind for the lottery protocol with independent ticket 
winning events that satisfies the $1 - \epsilon$ coverage rule. This computation 
is done using a probabilistic analysis that relies on modeling 
the payment process in \sysname under this protocol. In what follows, we state 
and prove a formula to calculate this balance.

\begin{theorem}
For an escrow with lifetime \lesc rounds, ticket issue rate 
\tktrate, lottery winning probability $p$, winning ticket currency value $\beta$, 
and parameter $\epsilon$, where $l_{esc}, tkt_{rate} \in \mathbb{N}$, 
$\beta \in \mathbb{R}^+$, and $0 \leq p, \epsilon \leq 1$, the value of 
\bescind needed to cover all 
winning lottery tickets with probability at least $1-\epsilon$ under \sysname setup 
with independent ticket winning events is given by:
\begin{equation}\label{payment-balance-ind}
B_{\emph{escrow}}^{\emph{ind}} = \beta F^{-1}(p,  l_{\emph{esc}} tkt_{\emph{rate}}, 1-\epsilon)
\end{equation}   
\noindent where $F^{-1}$ is the inverse 
of the cumulative distribution function (CDF) of the binomial 
distribution parameterized by $p$ and a number of trials 
$t =  l_{\emph{esc}} tkt_{\emph{rate}}$ at the value $1-\epsilon$. 
\end{theorem}

\begin{proof}
Given that we work in the random oracle model, and that 
we model the block hashes on the blockchain as a 
uniform distribution, lottery winning events are independent and can be 
modeled as Bernoulli trials. This means that the total number of 
winning tickets tied to an escrow with lifetime \lesc rounds, ticket issue rate 
\tktrate, and lottery winning probability $p$, is a binomial random variable 
parameterized by $p$ and a number of trials 
$t =  l_{\emph{esc}} tkt_{\emph{rate}}$.

Requiring \bescind to cover all winning tickets with probability $1-\epsilon$ 
means that \bescind must contain sufficient currency to pay a number of 
winning tickets $\psi$ that hits the $(1-\epsilon)$th percentile of the above binomial  
distribution, i.e., \bescind = $\psi\beta$. This number can be computed 
as:
\begin{equation}
\psi =  F^{-1}(p,  l_{\emph{esc}} tkt_{\emph{rate}}, 1-\epsilon)
\end{equation}

\noindent where $F^{-1}$ is the inverse 
of the cumulative distribution function (CDF) of the binomial 
distribution parameterized by $p$ and a number of trials 
$t =  l_{\emph{esc}} tkt_{\emph{rate}}$ at the value $1-\epsilon$. 
Substituting this expression in \bescind = $\psi\beta$ produces the formula 
stated in the theorem above, which completes the proof. \qed
\end{proof}

\subsection{Computing a Lower Bound for \bpenind}
\label{penalty-analysis-ind}
\vspace{-8pt}
In this section, we compute a lower bound for the penalty  
deposit \bpenind required to deter cheating under the lottery 
protocol with independent 
ticket winning events. This is done using a similar approach to the 
one used in Appendix~\ref{penalty-analysis}, as it applies the 
same utility function definition and covered malicious strategies. 
In what follows, we present this analysis, which includes stating and proving a 
lower bound for \bpenind. \\

\noindent{\bf Game setup.}
Similar to before, we have a single player game in which a malicious customer 
applies the ticket duplication strategy. However, we do not have the concept of fat 
rounds, where a round in this protocol is the exact time needed a single block on 
the blockchain. This means that the cheating detection period is \ddraw + \dredeem 
rounds. Furthermore, given that tickets enter the lottery independently of each 
other, there is a chance that none of the tickets will win. Thus, duplicating all tickets in 
a round is an option for a rational customer. For this reason, in each round of this 
game, a malicious customer
may decide to duplicate $y_i$ sequence numbers, such that 
$y_i \in \{1, \dots, tkt_{\emph{rate}} \}$.

Table~\ref{notations-2} summarizes the new notations we use in 
this section, beside redefining a few that appeared in Table~\ref{notations}, 
to describe this game. \\

\begin{table}[t!]
\caption{Notations II.} 
\label{notations-2}
\centering 
\small{
\begin{tabular}{| p{0.1\columnwidth} | p{0.82\columnwidth} |}\hline\hline
{\bf Symbol} & {\bf Meaning}  \\ [0.5ex] \hline\hline

$\kappa$ & The number of tickets that can be issued per round, such that $\tau  =$\tktrate and $\tau \in \mathbb{N}$. \\ [0.5ex]  \hline
$w$ & The lottery draw period in rounds, such that $w =$ \ddraw and $d \in \mathbb{N}$. \\ [0.5ex]  \hline
$x$ & The ticket redemption period in rounds, such that $x=$\dredeem and $x \in \mathbb{N}$. \\ [0.5ex]  \hline
$y_i$ & Number of duplicated tickets in round $i$, such that $0 \leq y_i \leq \kappa$. \\ [0.5ex]  \hline
$\upsilon$ & The escrow lifetime in rounds, such that $\upsilon=$\lesc and $\upsilon \in \mathbb{N}$. \\ [0.5ex]  \hline

\end{tabular}}
\vspace{-10 pt}
\end{table}

\noindent{\bf Additional utility gain analysis.}
We now state and prove a lower bound for \bpenind based on the 
above game setup.

\begin{theorem}
For the game and escrow setup described above, 
issuing invalid or duplicated lottery tickets is less profitable   
in expectation than acting in an honest way if: 
\begin{equation}\label{penalty-bound-ind}
B_{\emph{penalty}}^{\emph{ind}} > (m-1)p\beta\kappa \bigg(\frac{1}{1-(1-p)^{\kappa}} + w+x-1 \bigg)
\end{equation}

\end{theorem}

\begin{proof}
In \sysname, a customer can create an 
escrow with an $\upsilon$ round lifetime. All tickets issued in a round enter the lottery 
after $w$ rounds, and all winning tickets will expire after $x$ rounds from their  
lottery draw time. 

During each round of an escrow lifetime, an honest  
customer can issue up to $\kappa$ tickets with unique 
sequence numbers. As before, each ticket has an expected value of 
$p\beta$ coins, which corresponds 
to the service value a customer obtains for handing out this ticket.

On the other hand, for each round 
$i \in \{1, \dots, \upsilon\}$, a malicious customer would 
decide to duplicate $y_i$ tickets, where $y_i \in \{1, \dots, \kappa\}$. If none of the  
duplicated tickets win, which happens with probability $(1-p)^{y_i}$, this customer 
stays in the system and obtains an additional utility gain of $(m-1)p\beta y_i$ 
over what an honest customer would obtain. On the other hand, if 
any of these tickets win the lottery at round $i + w$, which happens with probability 
$1 - (1-p)^{y_i}$, the customer will be detected at round $i+w+x$ (the latest). This 
reduces its additional utility by \bpenind since the miners will revoke its penalty escrow.
At this time, the malicious customer still has $x$ rounds to issue tickets from the time 
of learning that it will be caught. Therefore, this customer will choose to duplicate 
all tickets in these rounds as described before.

Similar to Section~\ref{penalty-analysis}, we use a 
decision process diagram that captures the decisions a 
malicious customer makes over an escrow lifetime. 
As an example, we consider a simple case where we have an 
escrow with a 3 round lifetime, $w = 2$ rounds, 
and $x = 1$ round. The decision process for this setup is captured in   
Figure~\ref{penalty-decision-d2-z1-ind}.

\begin{figure}[t!]
\centerline{
\includegraphics[height= 1.4in, width = 1.0\columnwidth]{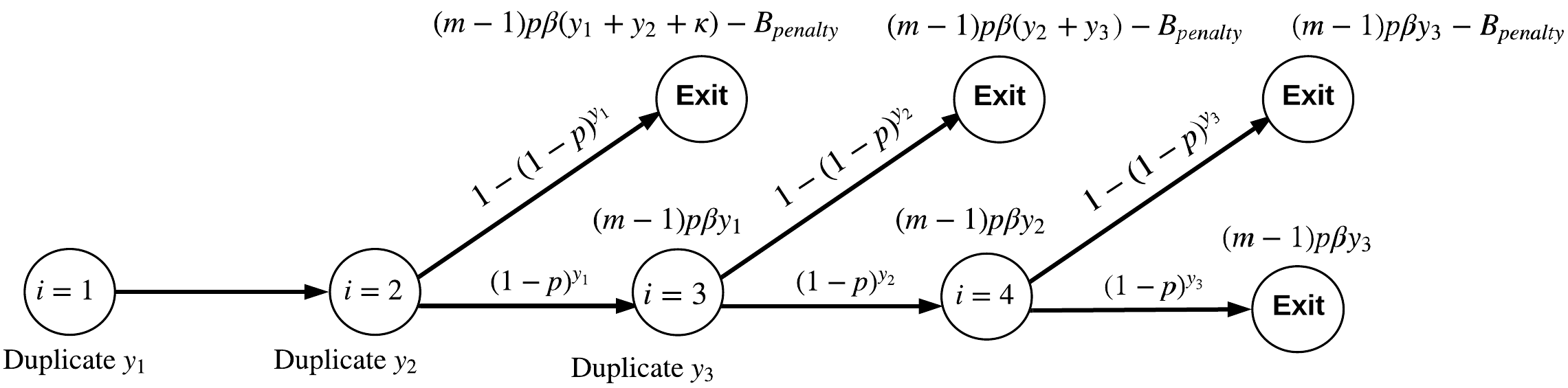}}
\vspace{-7pt}
\caption{Decision process for a $3$ round escrow with $w=2$ rounds and $x=1$ round. Arrows 
carry probabilities, decisions are found below the states, and 
the utility gain is found above the states.} \label{penalty-decision-d2-z1-ind}
\end{figure}

As shown, the only difference from what we saw in Figure~\ref{penalty-decision-d2-z1}, 
beside dealing with normal rounds, is the cheating detection probability. So, 
this process follows the same logic for deriving the additional utility at each round 
based on the likelihood that the cheating will be caught.

We utilize the same recursive idea in analyzing the above decision process. 
During the first round of an $\upsilon$ 
round escrow, a malicious customer will decide to duplicate $y_1$ 
tickets. If any of these tickets wins at round $1 + w$, cheating will 
be detected. Thus, the customer 
will duplicate all tickets for the next $x$ rounds and will pay the 
penalty \bpenind. This means that with probability $1 - (1-p)^{y_1}$, the 
utility gain is 
$(m-1)p\beta\big(\sum_{i=1}^w y_i + x\kappa\big) - B_{\emph{penalty}}^{\emph{ind}}$.

If none of the duplicated $y_1$ tickets wins the lottery, the customer stays in the system. 
This means that with probability $(1-p)^{y_1}$, the utility gain of this customer will be 
$(m-1)p\beta y_1 + \mathbb{E}_{\upsilon-1}[u(\mathcal{\hat{C}})]$, 
where the 
second term denotes the expected utility gain of a malicious customer 
in an $\upsilon-1$ round escrow.

Based on the above, we can express $\mathbb{E}_{\upsilon}[u(\mathcal{\hat{C}})]$  
as follows:
\begin{multline}
\mathbb{E}_\upsilon[u(\mathcal{\hat{C}})] = \big(1-(1-p)^{y_1} \big)\bigg((m-1)p\beta \sum_{i=1}^w y_i + (m-1)p\beta x\kappa - B_{\emph{penalty}}^{\emph{ind}}\bigg) + \\ (1-p)^{y_1}\bigg( (m-1)p\beta y_1 + \mathbb{E}_{\upsilon-1}[u(\mathcal{\hat{C}})] \bigg)
\end{multline}

As before, we have $\mathbb{E}_{\upsilon-1}[u(\mathcal{\hat{C}})] < 0$ since the penalty for 
an $\upsilon-1$ round escrow has been configured in a way that makes 
cheating unprofitable. Hence, and 
by requiring $\mathbb{E}_{\upsilon}[u(\mathcal{\hat{C}})] < 0$ 
to deter cheating, we find that: 
\begin{equation}\label{temp-2}
B_{\emph{penalty}}(y_1, \dots, y_d) > (m-1)p\beta \bigg(\frac{y_1}{1-(1-p)^{y_1}} + \sum_{i=2}^w y_i + x\kappa \bigg)
\end{equation}

For any $w$ and $x$ value, the above quantity is maximized when  
$y_i = \kappa$ for $i \in \{1, \dots, w\}$.\footnote{This is done in a similar way to the one found in 
Section~\ref{penalty-analysis}.} Substituting this in equation~\ref{temp-2} produces the lower 
bound stated in the theorem, which completes the proof. \qed
\end{proof}

As an example, consider an escrow with a 200 round lifetime, 
$\tau = 1000$ tickets, $p = 0.01$, $\beta = 1$ coin, $m=5$, $d = 6$, 
$r=6$, and $\epsilon = 0.01$. Applying equation~\ref{payment-balance-ind} 
produces $B_{\emph{escrow}} = 2,104$ 
coins, and applying equation~\ref{penalty-bound-ind} 
produces $B_{\emph{penalty}} > 480$ coins. Comparing these values to the ones 
obtained for the same example in Section~\ref{penalty-analysis}, where $B_{\emph{escrow}} = 2,000$ 
coins and $B_{\emph{penalty}} > 477.6$, 
we realize that these values re very close. This is because this example uses a long lifetime 
escrow and a larger ticket issue rate, which makes these values as close as possible to the 
expected values. Consequently, for short escrow lifetime and lower ticket issue rate, the 
difference will more dramatic. That is, the bounds for a lottery protocol with an exact 
win rate will be much lower, reducing the collateral cost for the customer.

\end{appendices} 

% that's all folks
\end{document}